\documentclass[prd, twocolumn,letterpaper, amsmath, amssymb, preprintnumbers, showpacs, floatfix, nofootinbib, superscriptaddress]{revtex4-1}
\usepackage{mystyle_emacs}
\usepackage{theoremref}
\usepackage{placeins}
\usepackage{multirow}

\usepackage{graphicx}
\usepackage{xcolor}
\usepackage[labelformat=empty]{subcaption}

\newcommand\newsubcap[1]{\phantomcaption%
       \caption*{\textbf{#1}}}
\newcommand{\he}{\mathit{He}}
\newcommand{\Psibar}{\bar \Psi}
\newcommand{\cas}{\overset{a.s.}{\longrightarrow}}
\newcommand{\cip}{\overset{p}{\longrightarrow}}
\newcommand{\cid}{\overset{d}{\longrightarrow}}
\newcommand{\mom}{\text{\tiny{MoM}}}
\newcommand\capa{\textbf{(a)}\,}
\newcommand\capb{\textbf{(b)}\,}
\newcommand{\cq}{{\cal Q}}
\newcommand{\cp}{{\cal P}}
\newcommand{\ct}{{\cal T}}
\makeatletter
\let\sv@endpart\@endpart
\def\@endpart{\thispagestyle{empty}\sv@endpart}
\makeatother

\begin{document}

\title{Infinite Variance in Monte Carlo Sampling of Lattice Field Theories}
\author{Cagin Yunus}
\email{cyunus@mit.edu}
\affiliation{Center for Theoretical Physics,
Massachusetts Institute of Technology, Cambridge, MA 02139, USA}
\author{William Detmold}
\email{wdetmold@mit.edu}
\affiliation{Center for Theoretical Physics,
Massachusetts Institute of Technology, Cambridge, MA 02139, USA}
\affiliation{The NSF Institute for Artificial Intelligence and Fundamental Interactions}

\begin{abstract}
In Monte Carlo calculations of expectation values in lattice quantum field theories, the stochastic variance of the sampling procedure that is used defines the precision of the calculation for a fixed number of samples. 
If the variance of an estimator of a particular quantity is formally infinite, or in practice very large compared to the square of the mean, then that quantity can not be reliably estimated using the given sampling procedure.
There are multiple scenarios in which this occurs, including in Lattice Quantum Chromodynamics, and a particularly simple example is given by the Gross-Neveu model where Monte Carlo calculations involve the introduction of auxiliary bosonic variables through a Hubbard-Stratonovich (HS) transformation. 
Here, it is shown that the variances of HS estimators for classes of operators involving fermion fields  are divergent in this model and an even simpler zero-dimensional analogue. To correctly estimate these observables, two alternative sampling methods are proposed and numerically investigated.

\end{abstract}

\maketitle

\newpage

\section{Introduction}
Quantum field theories (QFTs) at strong coupling are interesting in many contexts in  particle, nuclear,  and  condensed matter physics, but in many cases can only be quantitatively investigated using numerical approaches. One such approach involves discretising the theory on a spacetime lattice with a Euclidean metric. The functional integrals corresponding to measurable quantities can then be approximated using an importance sampling Monte Carlo method. 
In such a calculation, the probability of sampling a given configuration of the field degrees of freedom is determined by the Euclidean action and, depending on the parameters in the action, it is possible that field configurations enter with probability weights arbitrarily close to zero. 
If this is the case, certain random variables (observables corresponding to field operators) will have arbitrary large (infinite) variance.
As will be discussed below, quantities with infinite variance in standard sampling algorithms occur in phenomenologically-relevant theories such as Quantum Chromodynamics (QCD) due to zero-modes of the lattice Dirac operator as well as in other contexts. A particularly clear example is provided by correlation functions constructed from large numbers of fermion fields as will be the focus of this work.\footnote{Observables with infinite variances in fermionic theories have been analysed using a different approach in Ref.~\cite{inf_var_mc}.}

In applying Monte Carlo methods to QFTs, the Central Limit Theorem (CLT) is used to construct confidence intervals for the expectation value (mean) of the random variable from the corresponding variance over the samples. However, a random variable with infinite variance does not satisfy the conditions for the CLT and the sample variance of such a random variable is not meaningful because it does not converge to a particular value with increasing sample size. Moreover, the CLT is valid only in the limit that  the sample size approaches infinity and hence similar deficiencies will appear for random variables with finite but very large variances compared to squares of their means.
Despite these issues, there are physically interesting quantities in QCD and other field theories that formally have finite mean but infinite variance under standard sampling methods. To address these cases, alternative sampling schemes are required for reliable Monte Carlo estimates.

In this work, two methods will be introduced to address specific occurrences of infinite variance. The first method is applicable in the context of fermionic lattice field theories that are typically approached using the (continuous) Hubbard-Stratonovich (HS) transformation such as theories whose actions involve powers of fermion bilinear operators. A class of discrete HS transformations is introduced which generate discrete auxiliary bosonic variables. The variance of an estimator constructed from these discrete bosonic variables  will then be manifestly finite although it may be still very large compared to the square of its mean. This discrete sampling scheme is investigated in a toy model and in the 2D Gross-Neveu (GN) model.
While the approach is seen to be useful in some contexts, it becomes impractical in the limit of large spacetime volumes in its current implementation. 
The second method that is considered is a sequential reweighting procedure that is suitable for analysis of non-negative stochastic variables. With this method, the mean of a such a non-negative bosonic variable with infinite variance can be written as a product of the means of the several non-negative random variables each having finite variance. This approach is also investigated in the toy model and in the 2D GN model but can be applied in more complicated theories.

The structure of this work is as follows. In Sec. \ref{sec:inf-var-euclid}, the way in which random variables with infinite variances arise in lattice calculations of field theories such as QCD is outlined as a motivation for  subsequent studies of related phenomena in simple models. In Sec. \ref{sec:stat-sampling},  the main statistical concepts that are used in our analysis are introduced and interpreted. In Sec. \ref{sec:simple-models}, simple  models are introduced that cleanly exhibit the features that lead to observables with infinite variance. In Sec. \ref{sec:discrete-hs}, a novel discrete Hubbard-Stratonovich transform is presented that provides estimators with manifestly finite variance. This method is tested for the toy models introduced in Sec. \ref{sec:simple-models}. In Sec. \ref{sec:reweighting}, a new reweighting method that can be applied to non-negative stochastic variables is also introduced and this method is then tested for the toy model introduced in Sec. \ref{sec:simple-models}. Finally, Sec. \ref{sec:summary} summarises the results of this work and provides an outlook for future directions of investigation.
A number of important statistical results that support our main analysis are proven in Appendix \ref{appendix:probability} while Appendices  \ref{appendix:mom} and  \ref{app:hermite} present further technical details.

\section{Infinite variance in Euclidean Field Theory}
\label{sec:inf-var-euclid}

One can construct illustrative examples of infinite variance in phenomenologically-relevant theories such as lattice QCD. In this case, the partition function is given by:
\bad 
Z &= \int \cald[U]\cald[\Psi \Psibar] e^{-S[U] -\Psibar D[U] \Psi} \label{Z-QCD} \\ 
&=\int  \cald[U] e^{-S[U]} \det(D[U]) \\ 
&=\int \cald[U] e^{-S[U]}\prod_{\lambda \in \s_{D[U]}} \lambda ,
\ead 
where $U$ represents the gauge field and $\Psi$ and $\Psibar$ represent the fermions.
Here $S[U]$ is the bosonic part of the action of lattice QCD, $D[U]$ is the $N_D \times N_D$ Dirac matrix, the determinant of which arises from integration of the fermion degrees of freedom, and $\s_{D[U]}$ is the spectrum of $D[U]$ which accounts for multiplicities of the eigenvalues. It is assumed that the Dirac matrix $D[U]$ is diagonalizable for each $U$ and can be expressed as
\be 
D[U] = Q_U \Lambda_U \inv Q_U ,
\ee
where $\Lambda_U$ is a diagonal matrix consisting of eigenvalues $\lambda^a\in \s_{D[U]}$ of $D[U]$, and  $Q_U$ is not necessarily unitary.
With this definition, the columns $v^{(a)}_U$ of $Q_U$ and the rows $(w^{(a)}_U)^T$ of $\inv Q_U$ are  the right and left eigenvectors of $D[U]$ respectively and satisfy 
\bad 
\sum_{i} \lp w_{U}^{(a)} \rp_i \lp   v_U^{(b)}\rp_i = \d^{ab}, 
\\ 
\sum_a \lp   v_U^{(a)}\rp_i  \lp w_{U}^{(a)} \rp_j  = \delta_{ij},
\label{eq:wv-rel}
\ead
where $a$ and $b$ label the eigenvalues and $i$ and $j$ index the components of the corresponding eigenvectors.
It must be noted that one can permute and (independently) scale the columns of $Q_U$ freely. Furthermore, $Q_U$ can not generically be chosen continuously in $U$ and consequently the quantities $\l_U^a$, $v^{(a)}_U$ and $w^{(a)}_U$ depend implicitly on the choice of $Q_U$. In terms of these quantities, the components of the inverse of the Dirac operator for field $U$ can be expressed as:
\be 
\inv D[U]_{ij} = \sum_{a=1}^{N_D} \ff{1}{\l^a_U} \lp   v_U^{(a)}\rp_i \lp w_{U}^{(a)} \rp_j . \label{eq:eigenexpansion}
\ee
For certain values of the couplings that define the theory, there may be  an ``exceptional configuration'', that is a bosonic field configuration $U^*$ such that, for simplicity, strictly one of the eigenvalues, $\l^* \in \s_{D[U^*]}$, vanishes. In what follows, the corresponding left and right eigenvectors of $U^*$ will be denoted by $\lp w^*\rp^T$ and $v^*$ respectively. If such exceptional configurations exist, it can be seen that the standard estimators of physical quantities, such as fermion propagators, diverge. To illustrate this, consider a fermion field bilinear denoted as
\bad
V^1_{ij} = \Psibar_i \Psi_j
\ead 
and choose a particular combination of these bilinears weighted by the left and right eigenvectors at the exceptional configuration
\bad
\co = \sum_{i,j} w_i^* v_j^* V^1_{ij}.
\label{eq:firstO}
\ead
After the fermions are integrated out, for each sample size $N_S$, a standard estimator for the expectation value of $V_{ij}^1$ in a Monte Carlo calculation is 
\begin{equation} 
\hat V^1_{ij} = \frac{1}{N_S}\sum_{t=1}^{N_S} \inv D[U_t]_{ij},
\label{eq:V^1-estimator}
\end{equation}
where  $U_t$ for $t\in\{1,\cds,N_S\}$ are assumed to be independently and identically generated samples. The corresponding estimator for $\co$ is
\bad 
\hat \co_{N_S} = \ff{1}{N_S}\sum_{i,j} w_i^* v_j^* \sum_{t=1}^{N_s}\inv{D}[U_t]_{ij}.
\ead
The mean of $\hat \co_{N_S}$ is given as:
\bad 
\ev{\hat \co_{N_S}} = \ff{1}{Z}\int \cald[U] e^{-S[U]} \det(D[U]) \sum_{i,j} w_i^* v_j^* \inv D[U]_{ij}.
\label{eq:O-estimator-mean}
\ead
As one of the eigenvalues, $\l^*$,  for the field configuration $U^*$ vanishes, the integration measure in Eq.~\eqref{eq:O-estimator-mean} is such that $U^*$ will have vanishing probability of being sampled and consequently  the singularity due to $\inv D[U^*]_{ij}$ will not cause the expectation value to diverge.

Nevertheless, configurations in a neighbourhood\footnote{Precisely, for every $\epsilon > 0$, one can find a neighborhood $\mathcal{N}$ of $U^*$ such that $\sup_{U \in \mathcal{N}} \det D[U] < \epsilon $.} of $U^\star$, which will be sampled with a very small frequency governed by $\det D[U]$, will make large individual contributions to the sample mean but the expectation value will remain finite as $\det[U] \inv D[U]_{ij}$ is polynomial in $U$. To examine ${\rm var}(\co_{N_S})$ we consider  $N_S = 1$ for simplicity, noting that ${\rm var}(\co_{N_S}) = \ff{1}{N_S} {\rm var}(\co_{N_S=1})$. The variance of $\hat \co_1$ is
\bad 
{\rm var} \lp \hat \co_{1} \rp &= \int \cald[U] e^{-S[U]} \det D[U] \abs{\hat \co_1[U] - \ev{\hat \co_1}}^2.
\ead
Since $\hat \co_1[U^*] = \inv{(\l^*)}$ by construction and it was assumed that $\l^* = 0$ is the only vanishing eigenvalue of $D[U^*]$, the variance of $\hat \co_1$ is divergent as $\lp \l^*\rp^{-2}\det[U^*]$ is divergent. It must be stressed that, in an actual Monte Carlo calculation, exceptional configurations will not be sampled so the sample variance will remain finite for any finite sample size, but will not be bounded from above as the sample size increases.

The above example of a single fermion propagator illustrates the way in which infinite variance manifests but is not of physical relevance. However, correlation functions involving hadrons and nuclei in a theory such as QCD involve many propagators that arise from products of $k$ fermion bilinears. In this context, it is useful to consider the more general product
\bad 
V^k_{\{i\},\{j\}} = \prod_{n=1}^k V^1_{i_n,j_n},
\ead 
where $\{i\} \equiv \{i_1,\cds,i_k\}$ and $\{j\}=\{j_1,\cds,j_k\}$ label                         the fermions that enter in an ordered manner. A family of estimators for $V^k_{\{i\},\{j\}}$ analogous to Eq.~\eqref{eq:V^1-estimator} for each $N_S$ is
\bad
\hat V^k_{N_S;\{i\},\{j\}} &= \ff{1}{N_S}\sum_{t=1}^{N_S} \sum_{\pi \in S_k} s_{\pi} 
\prod_{n=1}^k \inv D[U_t]_{i_n,j_{\pi(n)}}
\ead
unless $\{i\}$ and $\{j\}$ contain repeated indices in which case $\hat V^k_{\{i\},\{j\}} = 0$ due to the anti-commutativity of fermions. Here, $S_k$ is the symmetric permutation group of order $k$, and $s_\pi$ is the sign of permutation $\pi$. Again choosing $N_S = 1$, one observes:
\bad
\hat V^k_{1; \{i\},\{j\}} = \sum_{\pi \in S_k} s_{\pi} 
\prod_{n=1}^k \hat V^1_{1;i_n,j_{\pi(n)}}.
\ead
If $N_0>1$ eigenvalues of $D[U^\star]$ vanish, then  it suffices to focus on a product of $N_0$ fermion bilinears:
\bad 
{\cal R}_{N_0} = \prod_{s=1}^{N_0} \sum_{i,j=1}^{N_D} \lp w^\star _{s}\rp_i \lp v^\star_{s}\rp_j V^1_{i,j},
\ead
where $N_D$ is the size of the Dirac matrix and $\lp w^\star _{s}\rp_i$ and $\lp v^\star_s\rp_j$ are the left and right eigenvectors of $D[U^\star]$ respectively, with vanishing eigenvalues $\l^\star_{s}=0$, for $s\in\{1,\cds,N_0\}$. For the estimator
\bad
\hat {\cal R}_{N_0} = \sum_{\pi \in S_{N_0}} \!\! s_{\pi} \!\! \sum_{\substack{i_1 \neq \cds \neq i_{N_0} \\ j_1 \neq \cds \neq j_{N_0}}}
\prod_{s=1}^{N_0} \lp w^* _{s}\rp_{i_s} \lp v^*_{s}\rp_{j_s} \inv D[U_1]_{i_s,j_{\pi(s)}},
\ead
where the first sum is over permutations $\pi$ in the symmetric group $S_{N_0}$.
Using the same arguments as for $\hat{\cal O}_1$, it can be shown that  $\hat {\cal R}_{N_0}$ has infinite variance.

The above arguments illustrate how infinite variances of estimators of physically-relevant quantities can arise in Monte Carlo calculations of theories including lattice QCD. We note that, the situation is exacerbated in quenched QCD, where the fermion determinant is taken to be unity, or in partially-quenched or mixed action QCD, where the Dirac operators entering the measure and the observables are different. In these cases, fermionic  observables can have infinite expectation values. Since the fermion action is different in the measure and in defining observables, similar concerns will arise in partially-quenched or mixed-action QCD. Without knowing that an observable in such a theory is free of the problem illustrated above\footnote{For example, the massive overlap Dirac operator does not have zero eigenvalues.}, standard sampling methods result in estimates of observables whose statistical behaviours are not governed by the CLT at any sample 
size and are unreliable.

\section{Statistical Sampling}
\label{sec:stat-sampling}
In this section, important results for stochastic variables that will be needed in the following analysis are introduced. A  review of the relevant aspects of probability theory and proofs of the results presented here are provided in Appendix \ref{appendix:probability}.

\subsection{A Natural Indicator of Infinite Variance}

For a sequence of independent and identically distributed (i.i.d.) random variables $\{X_n\}$, a sequence of random variables $\{s_n\}$  can be defined such that $s_n = \ff{1}{n-1}\sum_{i=1}^n \lp X_i - \bb X_n \rp^2$ where $\bb X_n = \ff{1}{n}\sum_{i=1}^n X_i$. Each $s_n$ is an unbiased estimator of the variance of $\bb X_n$ when it is finite. The $n\to\infty$ behaviour of $s_n$ provides empirical evidence as to whether the system has a finite variance or not. 
In particular:
\begin{itemize}
    \item Let $\{X_n\}$ be a sequence of  i.i.d. random variables with finite variance $\s^2$. Then as $n\to\infty$, $s_n \cas \s^2$, where the notation ``almost surely'' ($a.s.$) is defined in Appendix \ref{appendix:probability}.
    \item Let $\{X_n\}$ to be a sequence of i.i.d. random variables with finite mean $\m$ and infinite variance. Then, for any given $\delta > 0$, the number of  random variables $s_n$ that satisfies $s_n > \delta$ is infinite $a.s.$. 
\end{itemize}
The former statement follows from the Strong Law of Large Numbers, while the latter statement is proven in Appendix \ref{sec:main-theorems} as \thref{thm:infinite-jump}.

\subsection{Empirical Bias of the Sample Average for Finite Systems with Exceptional Configurations}

In systems that contains exceptional configurations, the convergence of the sample mean to the mean is slow and it is not straightforward to estimate uncertainties as the sample variance does not converge. These issues  resurface as empirical biases in systems with finite configuration spaces with configurations that are sufficiently infrequently sampled.
To explore this, let $\O$ be a finite sample space with $\abs{\O}$ elements. To this space, we associate the $\s$-algebra $\calf = 2^\O$ that is the set of subsets of $\O$, and a family of probability distributions $P^t:\calf \to [0,1]$ for $t \in (0,1]$. Here, $t$ corresponds to a parameter describing the system from which the samples are drawn such as a coupling constant or a mass. For a finite system, the knowledge of $P^t(\{\o\})$ for all $\o \in \O$ completely determines $P^t:\calf \to [0,1]$ through the requirement $P^t(A\in \calf) = \sum_{\o \in A} P^t(\{\o\})$. Therefore, it is enough to consider $P^t(\{\o\})$ and for brevity we define $P^t(\o) \equiv P^t(\{\o\})$. In the following, it is  assumed that $P^t$ is continuous in the sense that $P^t(\o)$ is a continuous function of $t$ for $t \in (0,1]$ for all $\o \in \O$ and that $X^t$ is a non-negative random variable which is continuous in $t$ in the same sense. Furthermore, the set of exceptional configurations is defined as $E \subset \O$ such that $\lim_{t \to 0} P^t(\o) = 0$ and $\lim_{t \to 0} P^t(\o)X^t(\o) \neq 0$ for all $\o \in E$. An element $\o \in E$ is referred to as an exceptional configuration and it should be noted that this definition depends on the choice of $X$ implicitly.

The mean of $X^t$,  $\m_{X^t}$, can be written as a sum of contributions from the exceptional configurations and contributions from the non-exceptional configurations. 
\bad 
\m_{X^t} &= \sum_{\o \in \O} P^t(\o)X^t(\o) = \m^e_{X^t} + \m^{/\!\!\!e}_{X^t},
\ead
where
\bad
\m^e_{X^t} &= \sum_{\o \in E} P^t(\o)X^t(\o), \\
\m^{/\!\!\!e}_{X^t} &= \sum_{\o \in E^c} P^t(\o)X^t(\o),
\ead
and $E^c=\Omega\setminus E$.

For a Monte Carlo estimate of the mean $\m_{X^t}$ with a fixed sample size $N_S$, the contribution from the exceptional configurations will be missing for $t$ sufficiently close to $0$, resulting in a  ``gap''  denoted by $\Delta_X \equiv \lim_{t\to 0} \m^e_{X^t}$. That is, denoting the actual mean of the observable by $\m_X \equiv \lim_{t\to 0}\m_{X^t}$, the sample mean will underestimate this value by $\D_X$ for ensembles that are large but not sufficiently large that the CLT applies, as will be discussed below.

Consider the product space $\O^{N_S}$ corresponding to the set of all ensembles of size $N_S$, that is every element $\o^{[N_S]} \in \O^{N_S}$ will correspond to a sequence of elements from $\O$: $\o^{[N_S]} = \left\{ \o^{[N_S]}_1,\cds,\o^{[N_S]}_{N_S}\right\}$. A new random variable which should be interpreted as the ensemble average for each ensemble can be defined by $\bb X^t_{N_S}\lp \o^{[N_S]}\rp = \ff{1}{N_S}\sum_{i=1}^{N_S}X^t \lp \o^{[N_S]}_i \rp$.

 Now let $p^{e}_{\min} (t) = \min_{\o \in E} P^t(\o)$ and $p^{/\!\!\!e}_{\min} (t) = \min_{\o \in E^\mathsf{c}} P^t(\o)$. As $t \to 0$, $p^e_{\min} (t) \to 0$ while $p^{/\!\!\!e}_{\min} (t) \not \to 0$. Therefore, for small enough $t$ one will have $[p^{/\!\!\!e}_{\min} (t)]^{-1} \ll [p^e_{\min} (t)]^{-1}$. The Weak Law of Large Numbers (see Appendix \ref{appendix:probability}) implies that for $N_S \gg [p^e_{\min}(t)]^{-1}$, $\bb X_{N_S} \simeq \m_X$ with very high probability. However, for $[p^{/\!\!\!e}_{\min} (t)]^{-1} \ll N_S \ll [p^e_{\min} (t)]^{-1}$, $\bb X_{N_S} \simeq \m_X-\D_X$ with very high probability. 
 
 For practical purposes, these results can be summarized by saying that for small $t$, with very high probability, $\bb X_{N_S}^t$ first approaches to $\mu_X - \Delta_X$ and then eventually converges to $\mu_X$ as $N_S$ is further increased. 
The above statements are made precise and proven as \thref{thm:gap} in Appendix \ref{sec:main-theorems}. 
It should be noted that if $E$ includes more than one element, $\bb X_{N_S}$ may demonstrate a series of plateaus before eventually converging to $\mu_X$. Figure \ref{fig:cartoon} schematically demonstrates the expected behaviour.
\begin{figure}
    \centering
    \includegraphics[scale=0.2]{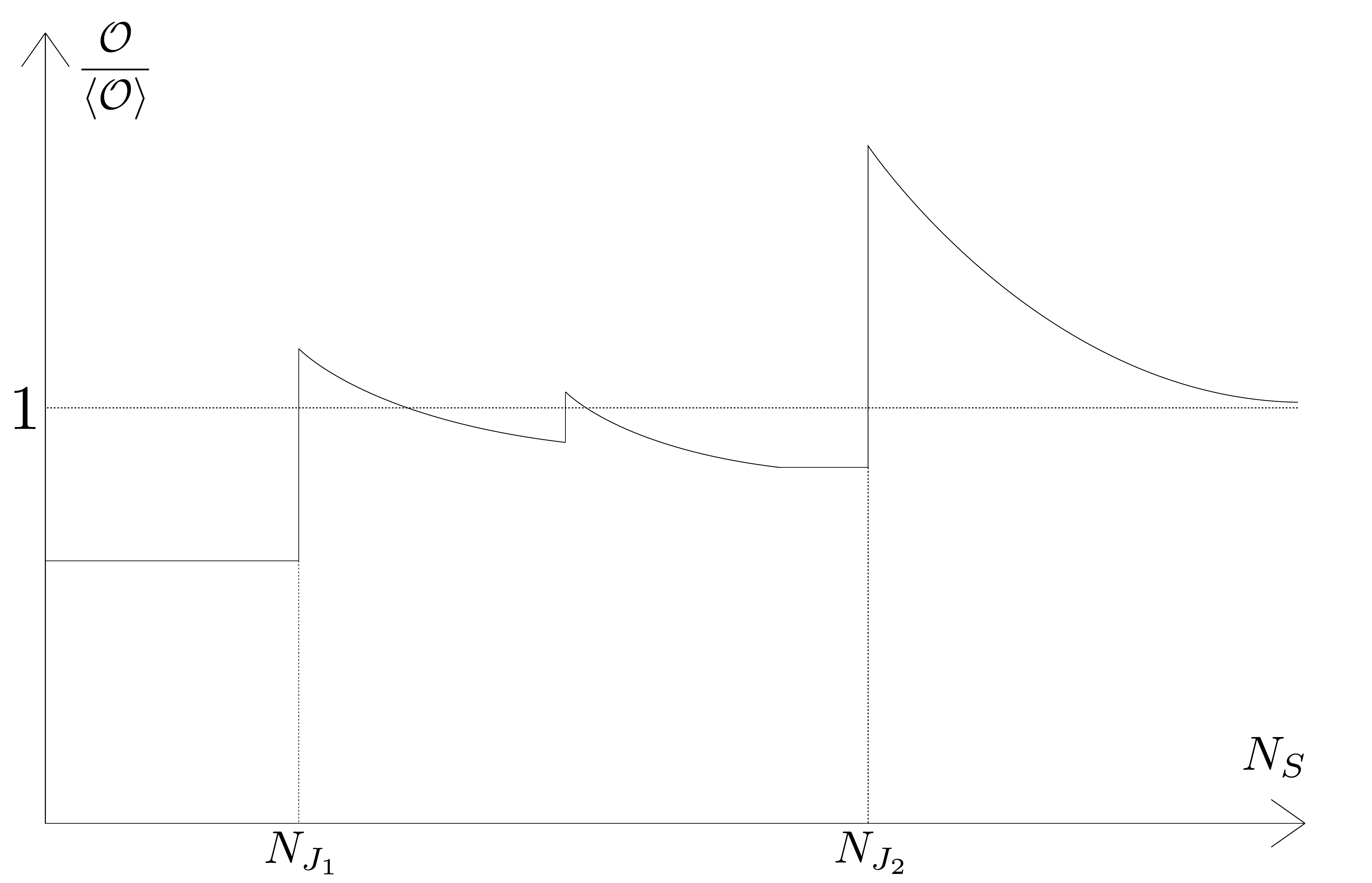}    
    \caption{The expected dependence on the number of samples of an observable with two exceptional configurations  under a particular sampling scheme. The  discontinuities at $N_{J_1}$ and $N_{J_2}$ correspond to the first time that the first and second exceptional configurations are sampled and these values can be arbitrarily large ($>10^8$ in examples below).
    }
    \label{fig:cartoon}
\end{figure}

\subsection{Non-Asymptotic Estimators}
\label{sec:non-asymptotic-estimator}
While the CLT is of utmost importance in statistical analysis, it is only valid asymptotically and for random variables with finite variance (see Appendix \ref{appendix:probability}). Therefore, the CLT is not applicable when dealing with random variables with infinite variance and the standard methods of estimation can not be utilized. Similar issues are also expected for a random variable with finite variance that has infinite variance in a certain limit, as such a variable is expected to be extremely non-Gaussian and require impractically large sample sizes for the CLT to apply.

To address these situations, non-asymptotic estimators are important, and
in this work the Median of Means (MoM) estimator will be used. The MoM is an estimator for which one is able to define confidence intervals which are also valid for random variables with infinite variance. After including the possibility of autocorrelations between samples, the MoM estimator can be defined as follows. 
Let $\{\hat \m_{1},\cds, \hat \m_{K}\}$ be the means of the random variable $X$ on each of $K$ independent batches of $B$ samples of $X$ obtained from the same stationary (thermalised) discrete time process. Then the median of means estimator $\hat \m_{\text{\tiny{MoM}}} = {\rm median}(\{\hat \m_1,\cds,\hat \m_K\})$.
Confidence intervals can be defined using
\bad 
Prob \lp \abs{\m_X- \hat \m_\mom} > 2 \s_X \ss{\ff{2\tt_{int,X}(B)}{B}} \rp \leq e^{-\ff{K}{8}}\,. \label{eq:mom_confidence}
\ead 
where $\m_X$ is the expectation value of $X$, $\s_X$ be the standard deviation of $X$, and $\tt_{int,X}(B)$ is the integrated autocorrelation time of the discrete time process.\footnote{$\tt_{int,X}(B)$ is defined  in the Appendix \ref{appendix:mom}.}
Further details and a proof of the above relation are provided in Appendix \ref{appendix:mom}.

\section{Simple examples with infinite variance}
\label{sec:simple-models}
In this section, two simple models are introduced and  exemplar correlation functions are investigated to illustrate the problem of infinite variance in Monte Carlo sampling. Numerical explorations of these models are presented in Secs. \ref{sec:discrete-hs} and \ref{sec:reweighting} below.

\subsection{Toy Model}
\label{sec:toy}
The first model considered is a zero dimensional (Euclidean)  theory of $2N_f$ interacting fermions represented  by\   $\Psi = \bpm \Psi_1,  \cds, \Psi_{2N_f} \epm^T$ and $\bar \Psi = \bpm \bar \Psi_1, \cds, \bar \Psi_{2N_f}\epm$ where $\Psi_i$ and $\bar \Psi_i$ are independent Grassmannian variables. The Lagrangian of this toy model is defined as
\begin{equation}
  \call = m\bar \Psi \Psi  - \ff{g}{2} \lp \bar \Psi \Psi \rp^2 ,
  \label{eq:toy-lag}
\end{equation}
where it is assumed that $g$ is positive. As shown in Ref.~\cite{original},  positivity of $g$ is required for the unitarity of realistic theories with four fermion interaction. 

The partition function of this theory coupled to (Grassmannian) sources $\bb\h$ and $\h$ is given by:
\begin{equation}
  Z[\h,\bb \h] = \int \prod_{i=1}^{2N_f} \lp d\Psi_i d\bb\Psi_i\rp e^{-m\bar \Psi \Psi  + \ff{g}{2} \lp \bar \Psi \Psi \rp^2 + \bb \h \Psi + \bb \Psi \h}.
\end{equation}
To calculate quantities derived from this partition function, one needs to remove the quartic term so that the Grassmannian integrations can be performed exactly. The standard way to do this is to introduce an auxiliary field through a (continuous) Hubbard-Stratonovich transformation \cite{hubbard,stratonovich}. It is straightforward to see that up to a multiplicative constant, the partition function is equivalent to
\begin{equation}
Z[\h,\bb \h] = \int_{-\infty}^\infty\!\! d\phi \int \prod_{i=1}^{2N_f} \lp d\Psi_i d\bb\Psi_i\rp e^{-\ff{1}{2}\phi^2 -(m+\ss g\phi)\bar \Psi \Psi + \bb \h \P + \bb \P \h},
\end{equation}
where $\phi$ is a real-valued scalar field.
The fermions can now be integrated exactly, leading to
\begin{equation}
Z[\h,\bb \h] = \int_{-\infty}^\infty d\phi\, e^{-\ff 1 2 \phi^2 + \bb \h \ff{1}{m+\ss g\f}\h} (m+\ss g \phi)^{2N_f}.
\end{equation}
Here, the Boltzmann weight 
\be
P(\phi) \propto e^{-\ff{1}{2}\phi^2}\lp m + \ss g \phi \rp^{2N_f}
\ee
is common to the partition functions and all correlation functions derived from it and therefore acts as the probability weight in importance-sampling Monte Carlo calculations. 

Now suppose that  one is interested in calculating the expectation value of the observable 
\begin{equation}
\co = \prod_{i=1}^{2N_f} \bar \Psi_i \Psi_i
\label{eq:eqOprod}
\end{equation}
which is determined by
\begin{equation}
  \lc \co \rc = \ff{1}{Z[0,0]}\lp \prod_{i=1}^{2N_f} \ff{\pp}{\pp \h_i}\ff{\pp}{\pp \bb \h_i} \rp Z[\h,\bb \h]\Bigg{\rvert}_{\h,\bb \h =0} \cd
\end{equation}
Using the auxiliary field, this is given by
\bad
  \lc \co \rc &= \ff{\int d\phi\, P(\phi) \lp m+\ss g\phi \rp^{-2N_f}}{\int d\phi P(\phi)} \\
  &= \ff{\int d\phi\, e^{-\ff 1 2 \phi^2}}{\int d\phi e^{-\ff{1}{2}\phi^2} \lp m+\ss g\phi \rp^{2N_f}},
\ead
which is clearly finite. 

Difficulties arise if this quantity is naively estimated through a Monte Carlo calculation. The standard estimator for this expectation is
\begin{equation}
\hat \co  = \ff{1}{N_S}\sum_{n=1}^{N_S} \tilde\co(\phi_n),
 \label{eq:continuous-estimator}
\end{equation}
where $N_S$ is the sample size and
\be
\tilde\co(\phi) = \lp m+ \ss g \phi \rp^{-2N_f}
\ee
is the representation of the observable in terms of the auxiliary field. 
This quantity has a singularity at $\phi^* = -\ff{m}{\ss g}$. While one will never sample this point because $P(\phi^*)=0$, with sufficiently many samples one will sample nearby points and they will cause large fluctuations in the estimation of the observable. In fact, the variance of this estimator is divergent, as the second moment (and all higher moments) of the bosonic operator $\tilde\co$ diverges:
\bad
\lc \co^2(\phi) \rc &= \ff{\int d\phi\, P(\phi) \lp m+\ss g\phi \rp^{-4N_f}}{\int d\phi P(\phi)} \\
&= \ff{\int d\phi\, e^{-\ff 1 2 \phi^2}\lp m+\ss g\phi \rp^{-2N_f}}{\int d\phi e^{-\ff{1}{2}\phi^2} \lp m+\ss g\phi \rp^{2N_f}} 
\\
&= \ii.
\ead


\subsection{Gross-Neveu Model}

To further explore the ideas  introduced above, it is useful to consider the $N_f$-flavour Gross-Neveu (GN) model \cite{original} which resembles QCD in a number of ways. In particular, it is asymptotically free and exhibits chiral symmetry breaking.\footnote{The version of the model introduced here has a discrete chiral symmetry but it is simple to modify the action to obtain a theory with a continuous chiral symmetry \cite{original}.}

Here, the Gross-Neveu model is defined in two dimensions on a discretised lattice geometry with Wilson fermions \cite{wilson-fermions}. Consider a rectangular lattice, described by the points $\{(s,t) | 1 \leq s \leq L,1 \leq t \leq T\}$ where $s$, $t$, $L$ and $T$ are positive integers and  lattice units are assumed throughout. Periodic (anti-periodic) boundary conditions are implemented in space (time). In this work, two-dimensional Dirac matrices are defined as
\begin{equation}\begin{aligned}
\g_0 &= \bpm 1 & 0 \\ 0 & -1 \epm, \qquad
\g_1 &= \bpm 0 & 1 \\ 1 & 0 \epm .
\end{aligned}\end{equation}
Denoting the masses by $m_i$ and the coupling constant by $g$, the partition function of the GN model is given by
\bad
Z =\int &\lp \prod_{s,t,i} d \bar \p_i d \p_i (s,t) \rp \\
&{\rm exp} \left\lbrace -\sum_{s,t,i} \bar \psi_i(s,t) K_i(s,t;s',t') \p_i(s',t') \right.
\\ &\qquad\qquad \left.+\ff{g}{2} \sum_{s,t} \lp \sum _i \bar \p_i(s,t)\p_i(s,t) \rp^2 \right\rbrace ,
\label{eq:ZGN}
\ead
where, $1 \leq i \leq N_f$ and
\begin{equation}
\begin{aligned}
K_i(s,t;s',t') =& \mathds{1}_{2 \times 2} \bigg( (2+m_i)\d_{s,s'}\d_{t,t'}  \\ &-\ff{1}{2}\big( \d_{s,s'+1}\d_{t,t'} + \d_{s,s'-1}\d_{t,t'} \\ &+ \d_{s,s'}\d_{t,t'+1}+ \d_{s,s'}\d_{t,t'-1} \big) \bigg) \\ &+ \ff{1}{2} \g_0 \big( \d_{s,s'}\d_{t,t'+1} - \d_{s,s'}\d_{t,t'-1}\big)  \\ &+\ff{1}{2} \g_1 \big(\d_{s,s'+1}\d_{t,t'} - \d_{s,s'-1}\d_{t,t'} \big).
\end{aligned}
\end{equation}
In the current work, $N_f=2$ flavours of fermions are used everywhere with $m_1=m_2=m$. By utilizing a Hubbard-Stratonovich transformation, the exponential in Eq.~\eqref{eq:ZGN} can be made bilinear in the fermion fields as in Sec.~\ref{sec:toy}. 
Indeed, the toy model in Sec.~\ref{sec:toy} is an approximation to the Gross-Neveu model in which the kinetic terms in the action are ignored.\footnote{In this approximation, Grassmannian variables at different sites are decoupled from each other and the GN model reduces to independent products of the toy model on each site.}

The set of exceptional configurations in the GN model is  more complicated than in the toy model discussed in the previous subsection. In particular, the exceptional configurations will correspond to a union of surfaces of codimension $1$ (and higher). 
For $L\times T=2\times 2$, the set of the exceptional configurations can be found algebraically by solving the characteristic equation of the Dirac operator for a given  set of parameters and is composed of two and three dimensional surfaces in the four-dimensional field-space. For larger lattice geometries, determination of these surfaces can in principle be performed numerically.

\section{Discrete Hubbard-Stratonovich Transformation}
\label{sec:discrete-hs}

The failure of sampling for some quantities with the standard HS transformation is tied to the continuous values taken by the auxiliary field, necessitating the existence of exceptional configurations in the models of the previous section. To avoid this, a family of discrete HS sampling schemes is introduced in this section and their utility in ameliorating the infinite variance problem is investigated numerically.

As introduced above, the continuous Hubbard-Stratonovich transformation is  given by
\begin{equation}
  e^{\ff 1 2 \chi^2} = \ff{1}{\ss{2\pi}} \int_{-\infty}^\infty du\, e^{-\ff 1 2 u^2 +  u \chi}.
  \label{eq:HS}
\end{equation}
This expression is valid for all commuting variables $\chi$. However, if $\chi$ is constructed out of fermion bilinears as in the models in Sec.~\ref{sec:simple-models}, Eq.~\eqref{eq:HS} need only  be satisfied up to terms ${\cal O}(\chi^{2N_f})$  (where $N_f$ is the number of fermions for the theories that have spinor dimension $2$) since higher powers of $\chi$  vanish identically.

To find additional solutions, solutions of
\begin{equation}
  e^{\ff 1 2 \chi^2} = \sum_{a\in{\cal A}} w_a e^{\xi_a \chi}
  \label{eq:discrete-cond}
\end{equation} 
are required,
where the index $a$ takes values in a finite  index set ${\cal A}$ that is to be determined.
The weights, $w_a$, are required to be non-negative to have a probabilistic representation and the $\xi_a$ are required to be real to avoid a sign problem. $\chi$ is assumed to satisfy $\chi^{2N_f+1} = 0$. 

After a change of variables $\chi \to i \chi$, solving the above equation is equivalent to solving
\begin{equation}
  e^{-\ff 1 2 \chi^2} = \sum_a w_a e^{i\xi_a \chi} +\co\lp \chi^{2N_f+1}\rp,
\end{equation}
where $\chi$ is considered as a real variable. That is, the above equation may be interpreted as the equality of the two real power series in $\chi$ up to the $2N_f$th order in $\chi$.
The series on the left and right sides of the above equation can be viewed as the characteristic functions\footnote{The characteristic function of a random variable $X$ is defined as $\f(\xi) = \lc e^{i \xi X} \rc.$} of two probability densities in a conjugate variable $\xi$, where these densities are 
\be
P_1(\xi)=\ff{1}{\ss{2\pi}}e^{-\ff 1 2 \xi^2}
\ee
and 
\be
P_2(\xi)=\sum_a w_a \d(\xi-\xi_a),
\ee
respectively. Eq.~\eqref{eq:discrete-cond} can thus be rephrased as finding a  polynomial $f(\xi)$ of degree at most $2N_f$ that  satisfies
\begin{equation}
 \ff{1}{\ss{2\pi}}\int_{-\ii}^{\ii} d\xi\, e^{-\ff 1 2 \xi^2}f(\xi) = \sum_a w_a f(\xi_a).
\end{equation}
Written in this form, the $\{\xi_a\}$ and $\{w_a\}$ can be found through the method of Gaussian quadrature. Denoting the Hermite polynomials by
\begin{equation} 
\he_n(\xi) = (-1)^n e^{\ff 1 2 \xi^2}\frac{d^n}{d\xi^n}e^{-\ff 1 2 \xi^2},
\end{equation}
the $N_f+1$ roots of $\he_{N_f+1}(\xi)$ give the $\xi_a$ and the $w_a$ are constructed as
\begin{equation}
  w_a = \ff{N_f!}{\he'_{N_f+1}(\xi_a)\he_{N_f}(\xi_a)}, \label{eq:hermite-w-def}
\end{equation}
as shown in Appendix \ref{app:hermite}. 

Having defined the sets $\{\xi_a\}$ and $\{w_a\}$, a Monte Carlo calculation can be performed for a Euclidean field theory as follows. Assume that the theory has a partition function:
\bad 
Z = \int \cald[U]\cald[\Psi\Psibar ] e^{-S[U]+\Psibar^{\a}_x D[U]_{\a x;\b y} \Psi^\b_y + \ff{1}{2} \sum_{x}\lp C_{\a\b} \Psibar^\a_{x} \Psi^\b_{x}\rp^2},
\ead 
where $\{\a,\b\}$ correspond to all fermion indices except the spacetime location $x=(s,t)$ and $C_{\a \b}$ is a complex matrix. If $C_{\a \b}\Psibar^\a_{\vec x,\tt}\Psi\b_{\vec x,\tt}$ is a sum of $k$ fermion bilinears, then $\lp C_{\a\b} \Psibar^\a_{\vec x,\tt} \Psi^\b_{\vec x,\tt}\rp^{k+1}$ will vanish. Then, the partition function can be expressed as:
\bad 
Z &= \int \cald[a]\cald[U]\cald[\Psi\Psibar] \prod_{x} w_{a_x} \\ &\qquad e^{-S[U]+\Psibar D[U] \Psi + \sum_{x} \xi_{a_{x}} C_{\a\b} \Psibar^\a_{x} \Psi^\b_{x}},
\ead
where $\cald[a] \equiv \prod_{x} \sum_{a_{x} \in \cala}$ and $\cala$ indexes the set of roots of $He_{N> k}$. Note that $N$ can be chosen to be any integer greater than $k$. 

After integrating over the fermion fields, one obtains:
\bad 
Z &= \int \cald[a]\cald[U]  e^{-S[U]}\det \lp D'[U,a]\rp\prod_{x} w_{a_{x}},
\ead
where $D'[U,a]_{\a,x;\a',x'}$ is given as
\bad 
D'[U,a]_{\a,x;\a',x'}= D[U]_{\a,x;\a',x'} + C_{\a\a'}\xi_{a_{x}}\d_{x,x'}.
\ead 
Consequently, one can  perform a Monte Carlo calculation using  $e^{-S[U]}\det \lp D'[U,a]\rp \prod_{x} w_{a_{x}}$ as the probability weight. 

The family of discrete HS transformations introduced here generalises the transformation first proposed by Hirsch~\cite{hirsch} and used extensively in the context of Quantum Monte Carlo simulations~\cite{Drut,Wu:2005zzb}. The form used in that work is equivalent to the $N_f=1$ case of the  transformation introduced above.

\subsection{Discrete Sampling vs.\ Continuous Sampling for the Toy Model}

In this section,  the toy model discussed in Sec.~\ref{sec:simple-models} is used to compare estimators based on  discrete HS transformations to each other and to the standard estimator based on the continuous HS transformation. The operator 
$\co = \prod_{i=1}^{2N_f}\bar \Psi_i \Psi_i$ in Eq.~\eqref{eq:eqOprod} 
combines fermion bilinears for each type of fermion in the model and provides a concrete example on which to focus. $N_f=2$ will be used in numerical studies.

The behaviour of the different estimators is determined by the model parameters $m$ and $g$ in Eq.~\eqref{eq:toy-lag}. The behaviour of the continuous estimator has been discussed above.
For the discrete HS-based estimators, the choice of $m$, $g$ and the order $N$ of the Hermite polynomial $He_N$ control the magnitude and probability of the least probable configuration. The roots and the corresponding weights for the first few Hermite polynomials are given in Table \ref{tab:Hermitelow}.
\begin{table}[!t]
\begin{tabular}{ |c|c|c| }
\hline
n & Roots $\xi_a^{(n)}$ & Weights $w_a^{(n)}$ \\ \hline
\multirow{2}{*}{2} & $-1$ & $1/2$ \\
 & $1$ & $1/2$ \\
 \hline
\multirow{3}{*}{3} & $-\ss 3$ & $1/6$ \\
 & $0$ & $2/3$ \\
 & $\ss 3$ & $1/6$ \\
\hline
\multirow{4}{*}{4} & $-\ss{3+\ss 6}$ & $1/12 \lp 3 - \ss 6 \rp$ \\
 & $-\ss{3-\ss{6}}$ & $1/12 \lp 3 +\ss 6 \rp$ \\
 & $\ss{3-\ss{6}}$ & $1/12 \lp 3 +\ss 6 \rp$ \\
 & $\ss{3+\ss 6}$ & $1/12 \lp 3 - \ss 6 \rp$ \\ 
\hline
\end{tabular}
\caption{\label{tab:Hermitelow}
Roots and weights of the $N\in\{2,3,4\}$ sampling schemes. Corresponding results for larger values of $N$ are given in Appendix \ref{app:hermite}.}
\end{table}

For the continuous HS estimator, Eq.~\eqref{eq:continuous-estimator}, samples are generated through the Metropolis-Hastings algorithm with the standard normal distribution chosen as the proposal distribution.
Discrete HS estimators are constructed for $He_N$ where $N \in\{3,\ldots,9\}$ with samples drawn through the Metropolis-Hastings algorithm with the weights given in Eq.~\eqref{eq:hermite-w-def} chosen as the proposal probabilities. For each sampling scheme, a total of $N_S=10^8$ samples are created for $m\in\{1.03,1.43,1.53,1.63,1.73,1.83,1.83,2.03,2.43\}$ and for $g = 1.0$. Autocorrelations are measured using the procedure of Ref.~\cite{wolff2004monte} and accounted for in the analysis. 

In what follows, the numerical data are analysed in $N_p=10^3$ steps by adding $10^5$ samples at each step. Precisely, at the step $k$, the samples that are included  are the set $\{1,\cds,k \cd 10^5\}$.  For each step  the data is analysed disregarding the samples not included and all metrics, including the autocorrelation times, are calculating independently for each step. 

In order to compare methods, the behaviours of the mean and the standard deviation of the continuous and discrete HS estimators are considered as a function of the sample size.
\begin{figure}[!t]
	\begin{subfigure}{1.0\linewidth}
    	\includegraphics[width=\columnwidth]{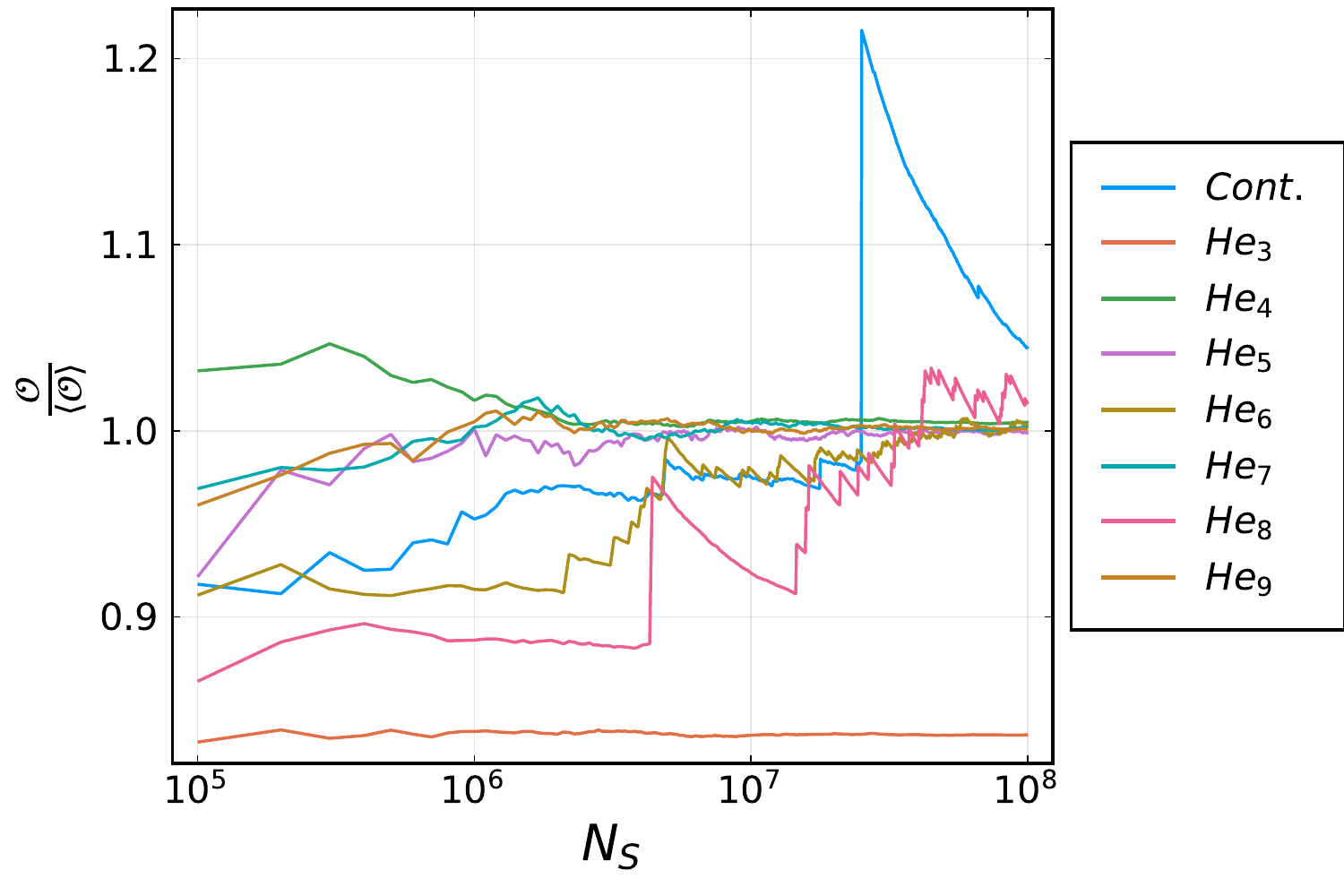}
    	\newsubcap{(a)}
    	\label{fig:toy_mean}
	\end{subfigure}
 \begin{subfigure}{1.0\linewidth}
	\includegraphics[width=\columnwidth]{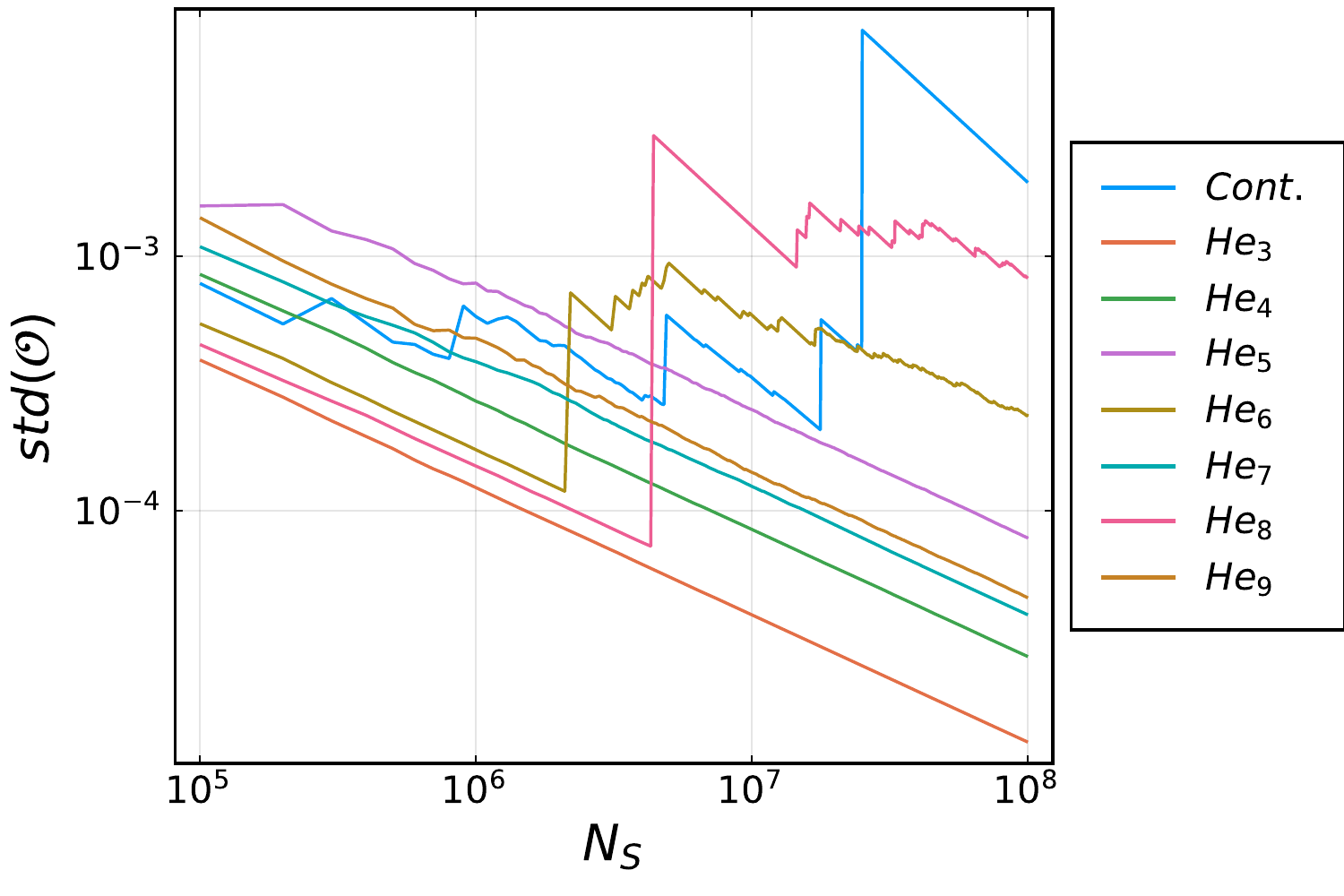}
	\newsubcap{(b)}
    \label{fig:toy_std}
	\end{subfigure}
	\caption{\capa shows the ratio of the sample mean of $\co$ in Eq.~\eqref{eq:eqOprod} to its exact value vs.\ sample size for $m=1.73$, $g = 1.0$ and $N_f =2$ with various sampling schemes for the toy model. \capb shows the standard deviation of $\co$, ${\rm std}(\co)$, as a function of the sample size for the same parameters. 
		\label{fig:toy}
		}
\end{figure}
Figure \ref{fig:toy} shows this comparison for $\he_N$ for $N \in\{3,\ldots,9\}$ at $m=1.73$ and $g=1.0$. These couplings are chosen such that the exceptional point $\phi^*=-m/\sqrt{g}$ is very close to one of the configurations in the $\he_3$ estimator ($\phi=-\sqrt{3}$). As can be seen from the behaviour of the mean, most of the discrete HS estimators rapidly converge to the exactly calculable value that is used to normalize the Monte-Carlo results. However, the continuous HS estimator shows significant jumps as the number of samples increases that occur whenever a sample sufficiently close to $\phi^*=-1.73$ is chosen, as expected from the general arguments in Sec. \ref{sec:stat-sampling}. Note that the binning of results in steps of $N_p=10^3$ has a smoothing effect on the mean; unbinned results show  more frequent and larger jumps. The $\he_3$ discrete sampling rapidly converges, but is biased even for $10^8$ samples. The $\he_8$ estimator also samples configurations close to $\phi^*$ (but not as close as for $He_3$) and correspondingly individual samples of these points significantly modify the mean, leading to the discontinuous jumps shown in the figure. The logarithm of the standard deviation shows the expected $1/\sqrt{N_S}$ behaviour for most of the discrete HS estimators, however the continuous HS estimator, and to some extent the $\he_8$ estimator, exhibits non-asymptotic scaling arising from samples close to $\phi^*$. As the number of samples increases, the continuous HS estimator will sample configurations arbitrarily close to $\phi^*$ and the non-asymptotic behaviour will persist indefinitely: the mean is not guaranteed to converge to the true value for any finite sample set and the variance will not monotonically decrease. This behaviour is anticipated by \thref{thm:infinite-jump} in Appendix \ref{appendix:probability} which shows that the large jumps observed in the variance will never cease.

The behaviour seen for the $\he_3$ and $\he_8$ estimators is in line with expectations given the configurations that are sampled and their respective probabilities. $\he_8$ has a root $t \simeq -1.63652$ that is close to the exceptional configuration $\phi^*=-1.73$, and consequently  the  $\he_8$ results show many jumps. This root of $\he_8$ is sampled with a probability $p \simeq 3\times 10^{-7}$ and is thus sampled about $30$ times for a sample size of $N_S = 10^8$. Supporting this expectation, it is observed that the first jump emerges around $N_S \sim \ff 1 p$ with the subsequent jumps are less marked.
For $\he_3$ discrete sampling, the variance is apparently behaving asymptotically, falling as $1/N_{S}$, despite the empirical bias observed in the mean. 
For this sampling, the root $t_a=-\sqrt{3}$ is sampled with probability $p \simeq 10^{-13}$. Since this root has not been chosen in the $N_S=10^8$ samples used in Fig.~\ref{fig:toy}, the mean is significantly underestimated. 
For $N_s \gtrsim 10^{13}$, the sample mean will begin to converge to the true value and the variance will exhibit jumps (as seen for $\he_8$).
For $N_S \gg 10^{13}$,  $t_a=-\sqrt{3}$ will be sampled representatively and the mean will converge to the correct value and the variance will decrease asymptotically.  While for this case the empirical bias would be observed with a very high probability if the same numerical experiments were repeated, it is not strictly a bias. With a very low probability the mean will be overestimated enormously making the estimator unbiased.

As this particular example shows, in the case of random variables with very large  variance,  asymptotic scaling of the variance is no guarantee of correctness. 
If the model parameters $m$ and $g$ are chosen such that the exceptional configuration is one of the roots of a given discrete HS sampling, the corresponding configuration will never be sampled, just as in the case of continuous HS sampling. Under these circumstances, the variance will decrease as $1/N_S$ but the mean will be biased. 

\begin{figure}[!t]
 	\begin{subfigure}{1.0\linewidth}
    	\includegraphics[width=\columnwidth]{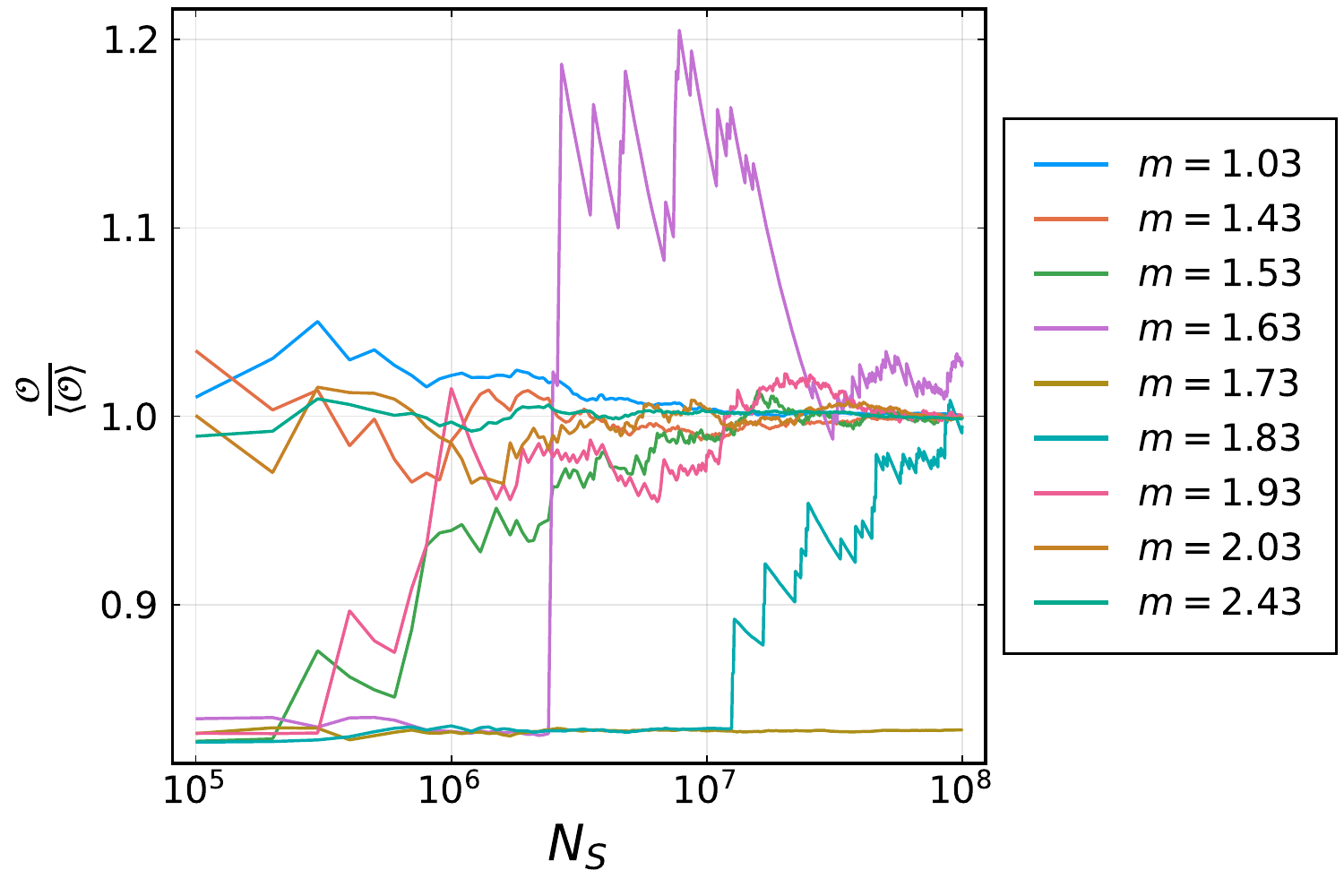}
    	\newsubcap{(a)}
    	\label{fig:hermite_3_mean_all}
	\end{subfigure}
	\begin{subfigure}{1.0\linewidth}
	\includegraphics[width=\columnwidth]{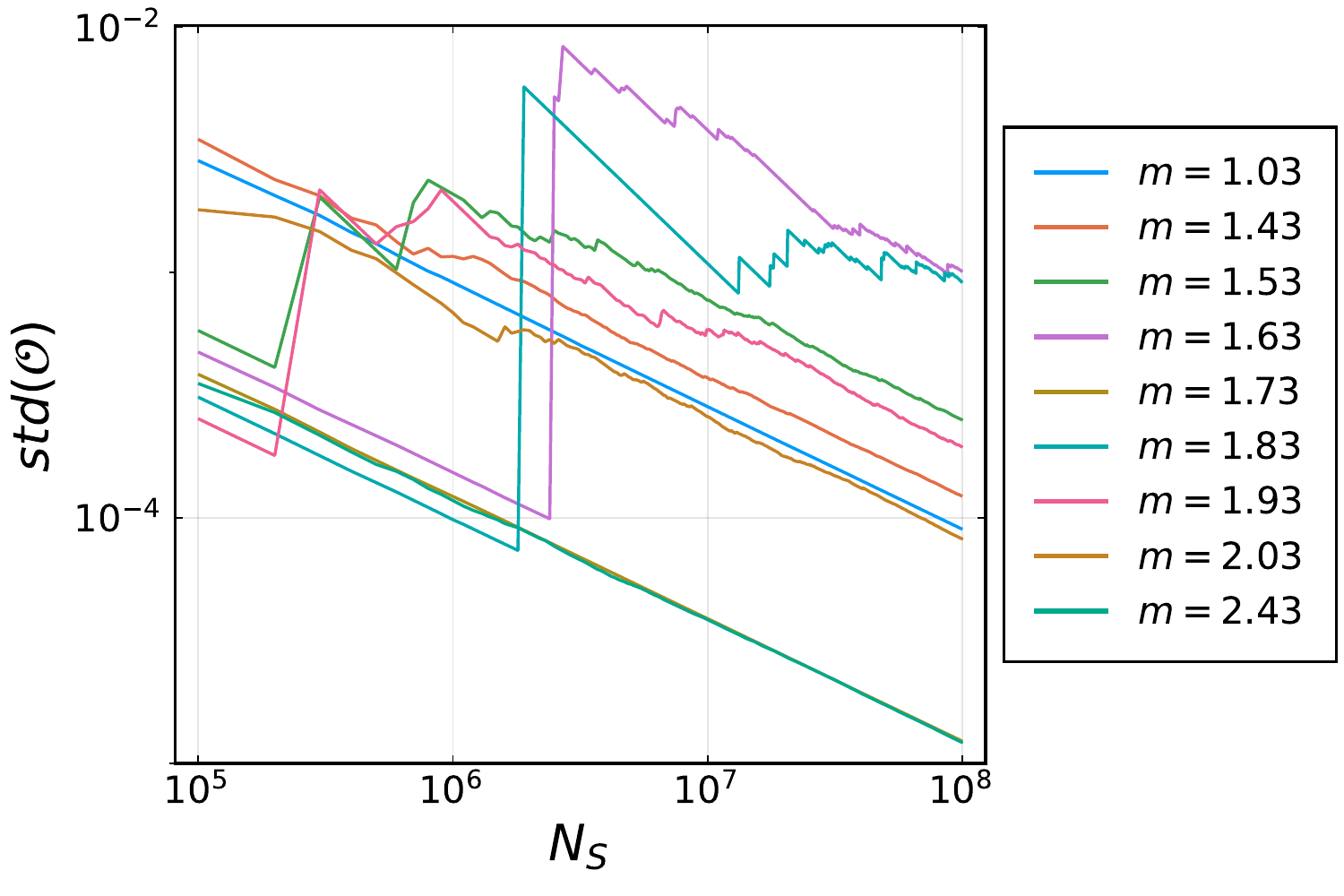}
	\newsubcap{(b)}
    \label{fig:he_3_std}
	\end{subfigure}
	\caption{\capa shows the ratio of the sample mean of $\co$ in Eq.~\eqref{eq:eqOprod} to the exact value  vs.\ sample size for various $m$ for $g = 1.0$ and $N_f =2$. The $\he_3$  discrete sampling scheme is used. \capb shows the standard deviation of $\co$ as a function of the sample size for the same parameters.
	\label{fig:2}}
\end {figure}

In Fig.~\ref{fig:2}, the mean and standard deviation of the same observable are studied for the $\he_3$ discrete HS sampling from $g=1.0$ and for a range of values of $m\in[1.03,2.43]$. As can be seen, for masses such that the exceptional value $\phi^*=-m/\sqrt{g}$ is not close to one of the roots $t_a\in\{-\sqrt{3},0,\sqrt{3}\}$, the calculations converge quickly to the correct value as the number of samples is increased and display the expected asymptotic $1/N_S$ scaling of the variance. However as $\phi^*$ moves closer to the root at $t_a=-\sqrt{3}$ from either above ($m=1.63$) or below ($m=1.83$), the convergence to the true value is much slower and large jumps are seen in the variance each time this root is sampled.
For $m=1.73$, the results apparently converge rapidly with $1/N_S$ scaling, but to an incorrect result at this number of samples (as in the previous figure). 

In fact, using the CLT, a sample size satisfying the conditions for $N(\d,\e)$ in Eq.~\eqref{eq:thmgap} can be found. For the current problem, $N(\d,\e)$ is a lower bound, such that for all $N \geq N(\d,\e)$ there is a range of mass values $m_l(N) \leq m \leq m_u(N)$ where $m_l(N) < \ss 3 < m_u(N)$ such that $P\lp \abs{\hat \co_{N} - (\m - \D)} \leq \d \rp \geq 1-\e$. To obtain a concrete value, we choose $\e = 3\times 10^{-7}$ corresponding to $5$ standard deviations for the standard normal distribution. Then, $N(\d,\e)$ can be chosen as $\lp \ff{\s_{X^0}}{\m-\D}\inv \Phi(1-\ff \e 2)\rp^2 \ff{1}{r^2} \approx \ff{59.07}{r^2}$ where $r = \ff{\d}{\m-\D}$ is ratio of the deviation $\d$ to the biased mean $\m-\D$ and $\Phi(x)$ is the cumulative distribution function of the standard normal distribution. For $r=0.01$, $N(\d,\e) = 6\times 10^5$  satisfies the required conditions.

To further investigate how the $\he_3$ discrete sampling behaves as the exceptional point of the theory moves towards one of the roots, the convergence of the sample average normalised to the true value is studied for $g=1.0$ and $m\in\{1.73,1.76,1.79,1.82\}$. In this simple toy model, the expected deviation of the sample mean  arises from the contribution of just one root that is the least probable and is straightforward to determine. 
Figure \ref{fig:he_3_transition} presents the results and shows that as the exceptional point approaches a root, the number of samples needed to remove the empirical bias increases, scaling approximately as the inverse probability of the least probable root. 

\begin{figure}[!t]
\includegraphics[width=\columnwidth]{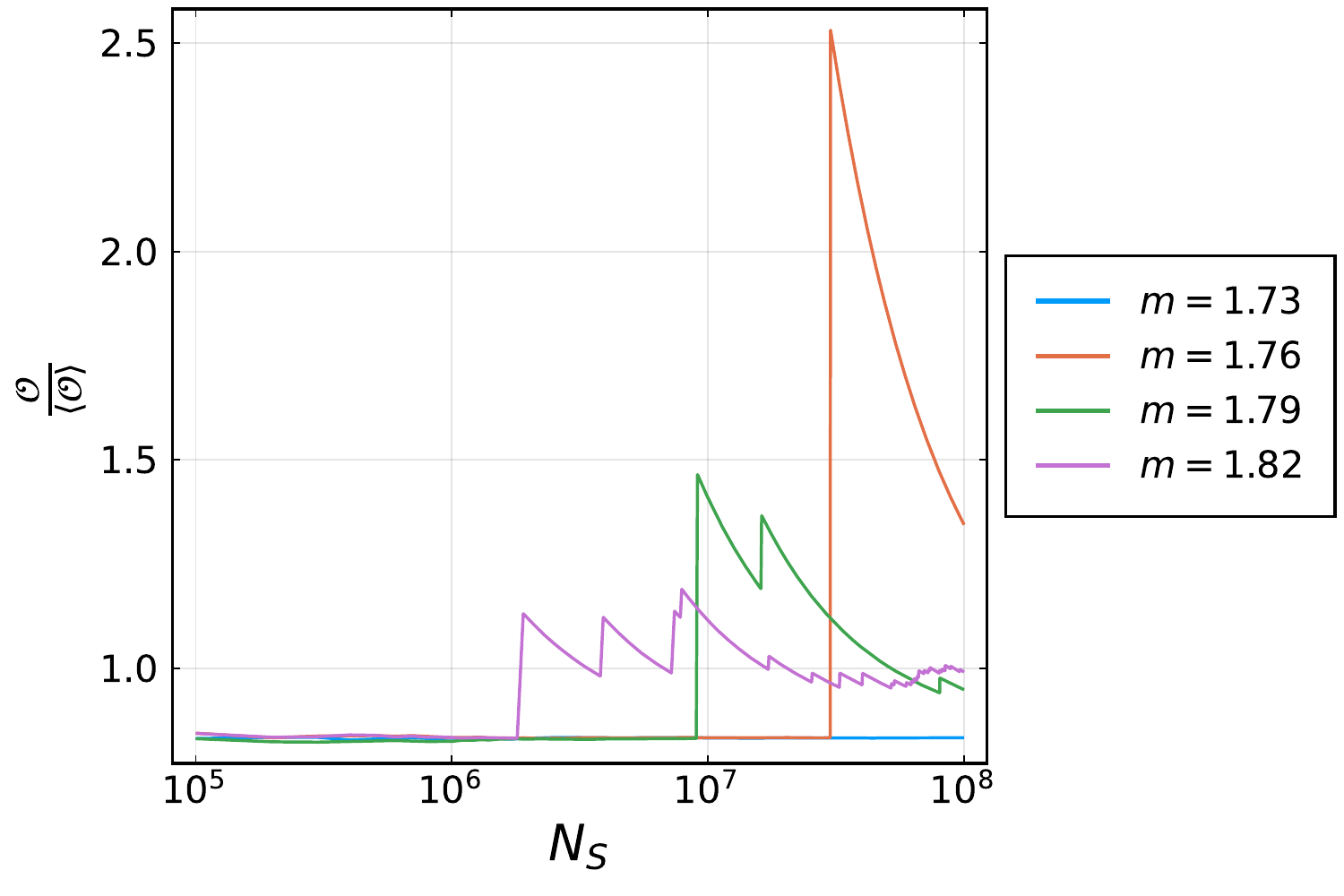}
\caption{The figure shows how (normalized) sample average moves over from $1-\ff{\D}{\m}$ to $1$ as $m$ deviates from $m=\ss 3$.
\label{fig:he_3_transition}}
\end{figure}

\subsection{Discrete Sampling for the Gross-Neveu model}
\label{sec:GN-discrete}

In this section, the effects of infinite variance are investigated in the context of the GN model. 
Calculations are undertaken for $N_f=2$ flavours of fermions and for various values of the fermion mass, $m$, and coupling, $g$. For a lattice of size $L\times T$ using the $\he_N$ discrete sampling there are $N^{LT}$ possible configurations.
\begin{figure}[!t]
\includegraphics[width=1.0\columnwidth]{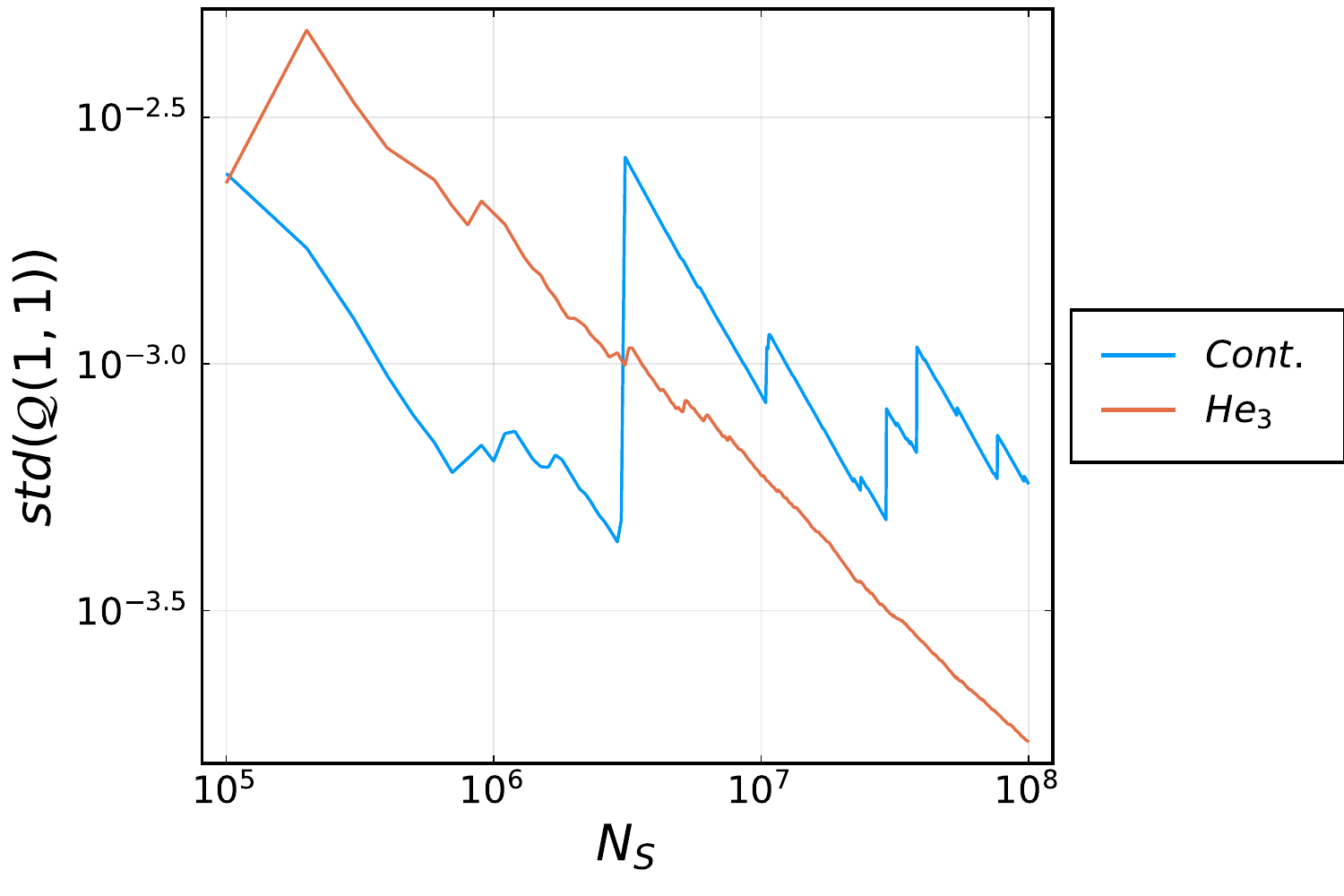}
\caption{Standard deviation of $\cq(1,1)$ vs.\ the sample size for a $2\times 2$ lattice, with $m=1.73$, $g = 1.0$ and $N_f =2$ using the $\he_3$ discrete sampling scheme for the Gross-Neveu model.}
\label{fig:gn_2x2_std}
\end{figure} 
As concrete examples, lattices of extent $L=T\in\{2,\ldots,8\}$ are investigated using the continuous and discrete $\he_3$ sampling schemes. 
For $L=T=2$, Fig.~\ref{fig:gn_2x2_std} shows the sample size dependence of the 
logarithm of the standard deviation of the observable
\be 
\cq(s,t) = \prod_{i=1}^{N_f} \prod_{\sigma=\uparrow,\downarrow}\Psibar_i^\sigma(s,t) \Psi_i^\sigma(s,t),
\label{eq:gn-op-def}
\ee
where the second product is over the fermion spin components. This quantity is evaluated at a single site, chosen to be $(s,t)=(1,1)$, and is constructed from all spin and flavour components of the fermion field at that site.\footnote{
Due to the spin-flavour symmetry of the model, this quantity involves a single eigenvalue entering with multiplicity $2N_f$. By translational symmetry, $\cq(s,t)$ is identical for any site.}
While the slope converges to $-0.5$ for the discrete sampling scheme, the standard deviation of the continuous scheme exhibits large jumps over the entire $N_S=10^8$ samples. 
\begin{figure}[!t]
\includegraphics[width=1.0\columnwidth]{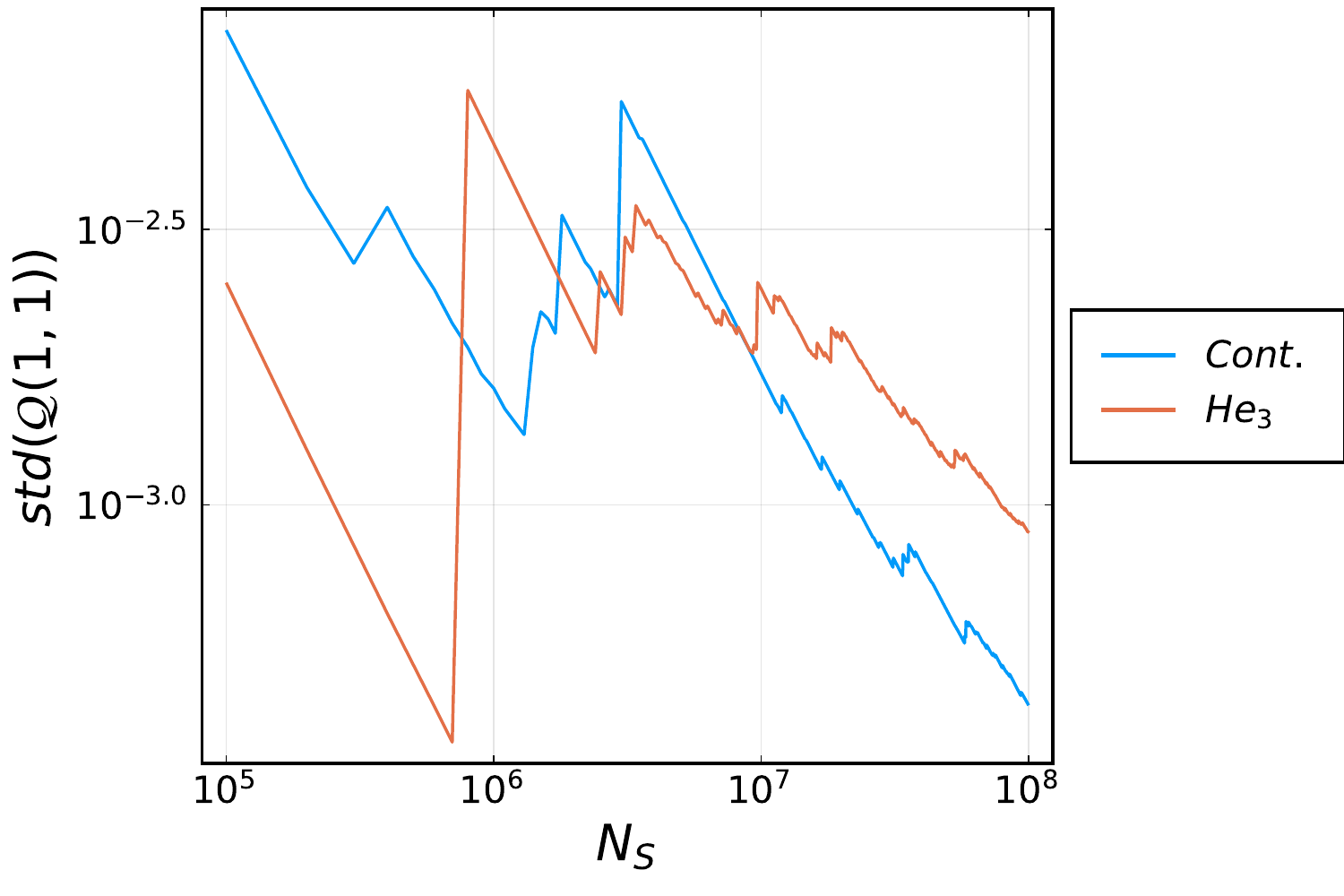}
\caption{Standard deviation of $\co = \mathcal{Q}(1,1)$ vs.\ the sample size for the $8\times 8$ lattice, with $m=1.73$, $g = 1.0$ and $N_f =2$ using the $\he_3$ discrete sampling scheme for the Gross-Neveu model. }
\label{fig:gn_8x8_std}
\end{figure} 

Fig.~\ref{fig:gn_8x8_std} displays results for the same quantities calculated using a larger lattice of extent $L=T=8$. It is clear that over the same range of sample sizes, even the discrete sampling scheme does not conclusively show the variance decreasing as $1/N_S$.

\begin{figure}[!t]
\begin{subfigure}{1.0\linewidth}
\includegraphics[width=0.9\columnwidth]{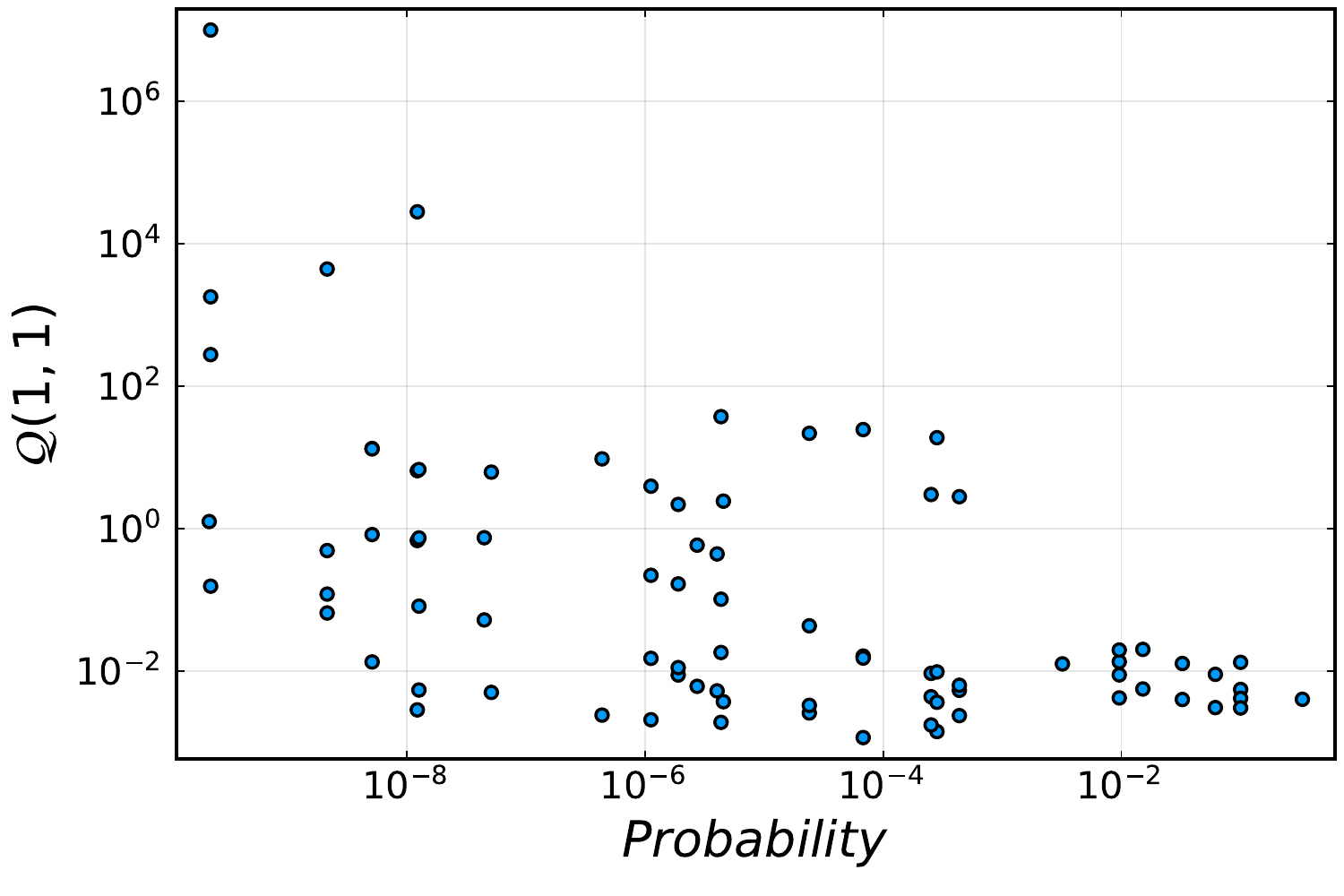}
\newsubcap{(a)}
\end{subfigure}
\begin{subfigure}{1.0\linewidth}
\includegraphics[width=0.9\columnwidth]{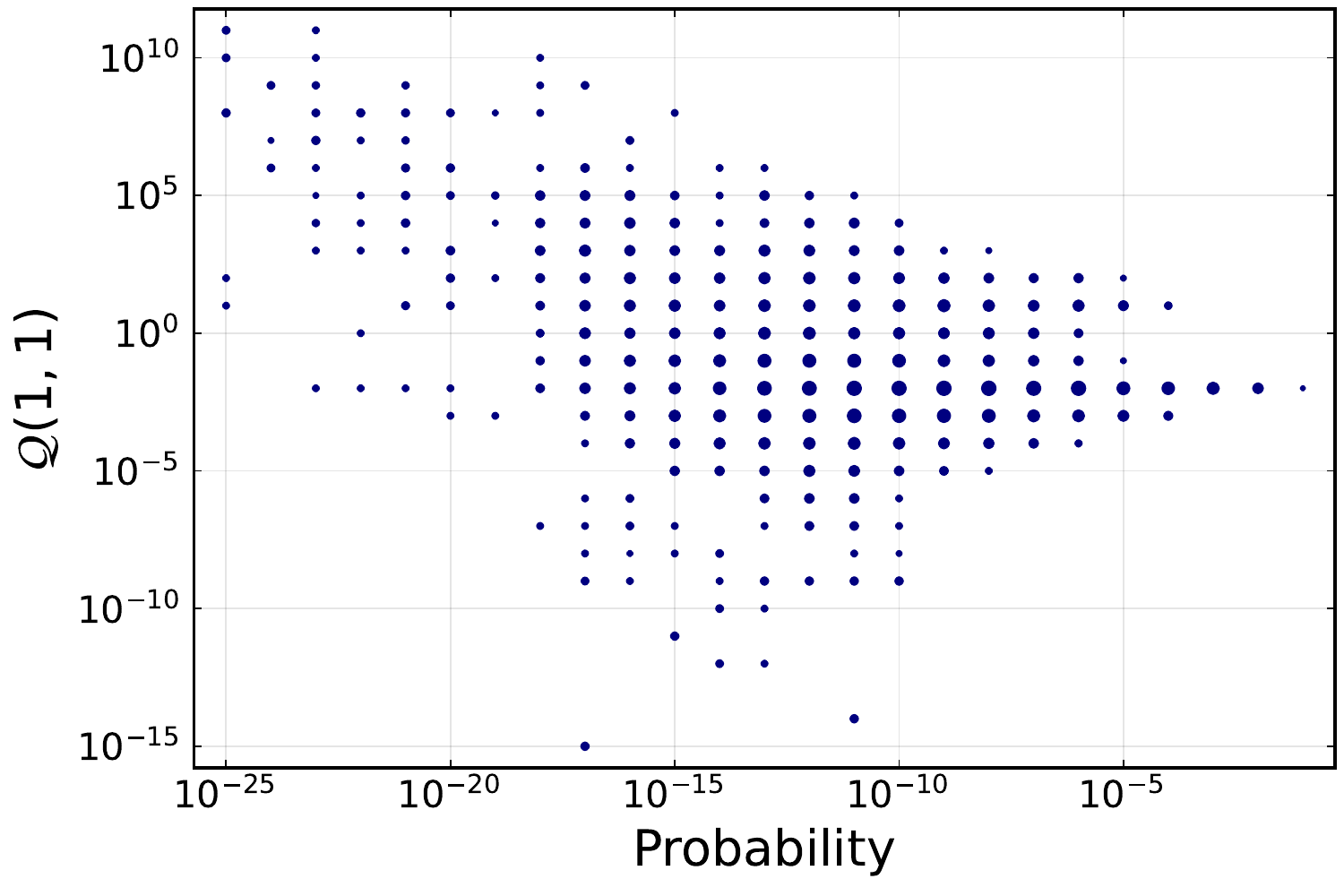}
\newsubcap{(b)}
\end{subfigure}
\begin{subfigure}{1.0\linewidth}
\includegraphics[width=0.9\columnwidth]{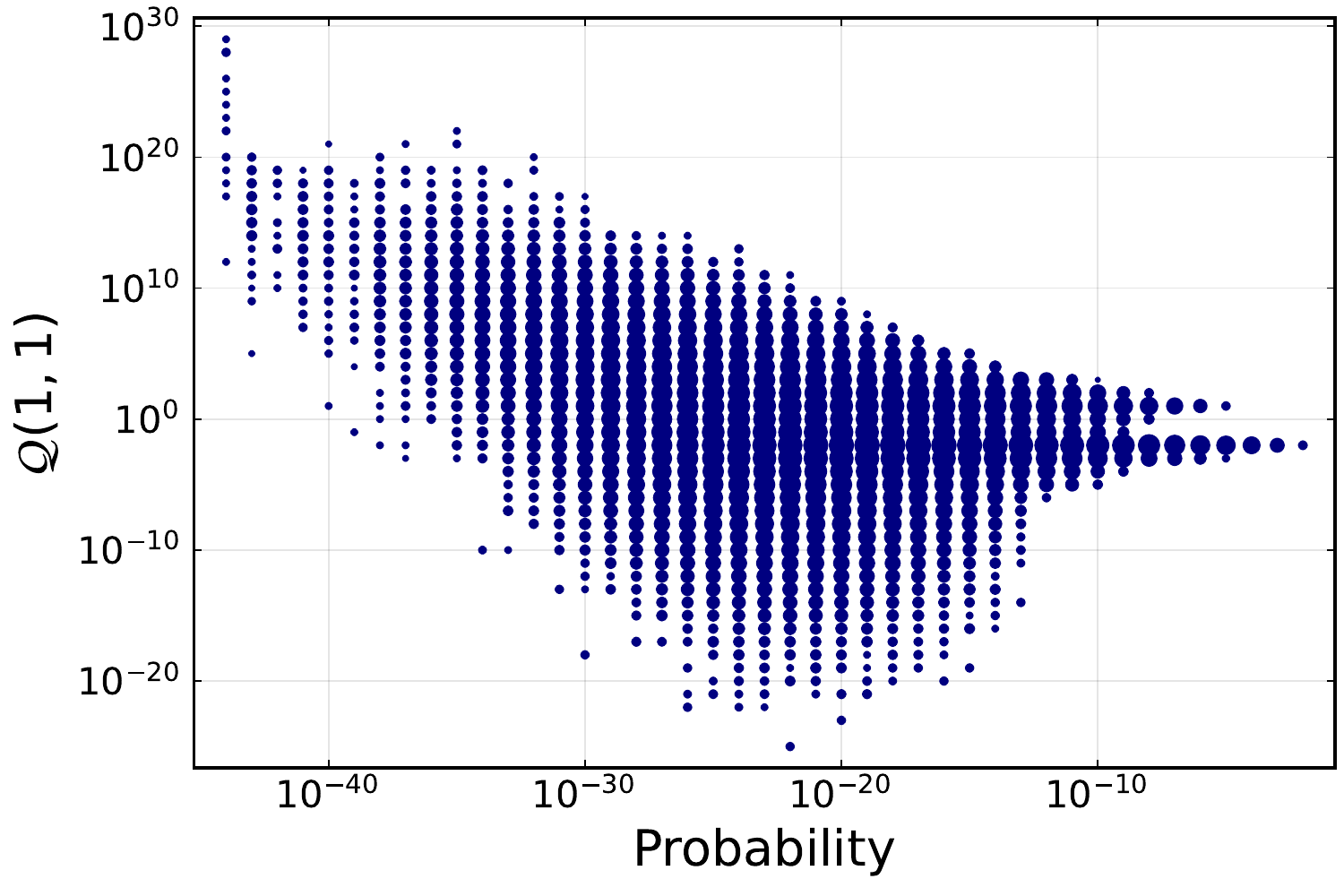}
\newsubcap{(c)}
\label{fig:log_count_local_4x4}
\end{subfigure}
\caption{The  values of $\cq(1,1)$ vs.\ the probabilities of the configurations for  (a) $2\times 2$ (b)  $3\times 3$  and (c) $4\times 4$ lattice geometries with $N_f=2$, $m=-1.5$, $\ss g =2.0$ and the $\he_3$ discrete sampling scheme for the GN model. For (b) and (c), the data are binned in units of one decade on both axes and the radius of the plot symbol indicates the number of samples in a given bin.
\label{fig:6}}
\end{figure}
The lack of convergence seen for the larger lattice 
can be understood by considering the spectrum of the logarithm of $\cq(1,1)$.
Fig.~\ref{fig:6} shows this spectrum on each configuration as a function of the logarithm of the probabilities of the configurations for  $L\times L$ lattices with $L\in\{2,3,4\}$ (since the number of configurations grows exponentially with $L$, results for $L>4$ are not shown). As can be seen in each case, there are a significant number of rare but important configurations. As $L$ increases, the number of these configurations increases rapidly.

The operator $\cq(1,1)$ is explicitly constructed such that for the continuous HS sampling scheme $P(\phi^*) = 0$ and  $\hat\cq(\phi^*) = \ii$ for at least one exceptional configuration $\phi^*$ while $\int_{\phi} P(\phi)\hat\cq(\phi) < \ii$ (here, $\hat\cq$ is the HS representation of $\cq(1,1)$ after fermions are integrated out). Since 
\be
P(\phi)\propto e^{-\frac{1}{2}\sum_x\phi^2(x)} \det[D(\phi)],
\ee
it follows that for $|\phi|<\ii$, $P(\phi^*) = 0$ occurs only when the determinant vanishes. 

For a valid discrete sampling scheme, $\phi^*$ will not be in the domain of the discrete variable $\xi$. However for $\xi$ close to $\phi^*$, 
\be
P(\xi)\propto w(\xi)\det [ D(\xi)] 
\ee
will be small since $w(\xi)>0$ and the determinant has the same functional dependence on either the continuous or discrete HS field. Similarly, $\hat\cq(\xi)$ will be large for $\xi$ near $\phi^*$ as both the continuous and discrete HS transforms result in the same functional form for $\hat\cq$ after the fermion fields are integrated out.
As a consequence of this behaviour,  configurations of smaller and smaller probabilities contribute larger and larger amounts to $\cq(1,1)$. 
From Fig.~\ref{fig:6},  it is clear that
this issue is exacerbated for larger lattices, Since the set of exceptional configurations grows with volume, the number of nearby configurations in discrete HS sampling with small probability and large contribution to $\cq(1,1)$ grows rapidly.
While the discrete sampling scheme $\cq(1,1)$ has finite variance, in practice one needs to have a sample size on the order of the inverse of the smallest probability to obtain a reliable estimate of $\ev{\cq(1,1)}$.
The smallest probability for a lattice with volume $V$ and number of degrees of freedom per site $c$ has an upper bound $\sim \co(c^{-V})$ although the smallest probabilities  will typically be much smaller.
Consequently for observables that have formally infinite variance, one needs to have a sample size that is greater than $\co(c^V)$ to properly estimate the mean.

As a comparison, Fig.~ \ref{fig:log-count-condensate} shows  the logarithm of the absolute value of an observable with finite variance, namely
\be
\overline\co = \ff{1}{L^2} \sum_{s,t,i,\sigma} \Psibar_i^\sigma(s,t) \Psi_i^\sigma(s,t)
\ee
for $L=4$, $N_f=2$, $m=-1.5$, $\ss g =2.0$ using the $\he_3$ sampling scheme.
In terms of the auxiliary variable, $\xi$, after integrating out the fermions, this operator will take a form $\xi \to \overline\co_{a.v.}(\xi)$, so that $\ev{\overline\co} = \sum_\xi P(\xi) \overline\co_{a.v.}(\xi)$. The notation $a.v.$ indicates the ``absolute value of the condensate'' which refers to the random variable $\xi \to \abs{\overline\co_{a.v.}(\xi)}$. Note that this definition depends on the particular auxiliary variable chosen.  In contrast to $\cq(1,1)$, $\overline\co$ only involves one fermion bilinear in each term in the sum and is thus less singular around exceptional configurations; although  for small probabilities  $\log \overline\co\sim -\log( \text{prob})$, this growth is not as severe as in the case of $\cq(1,1)$. This is made clear in Fig.~\ref{fig:8} which compares the behaviour of  $\overline\co$ and $\cq(1,1)$ directly. While the behaviour of observables with infinite variance in regions of low probability is controlled by the exceptional configurations, the more general structure of the log-count plots above is specific to the particular observable.
\begin{figure}[!t]
\includegraphics[width=0.9\columnwidth]{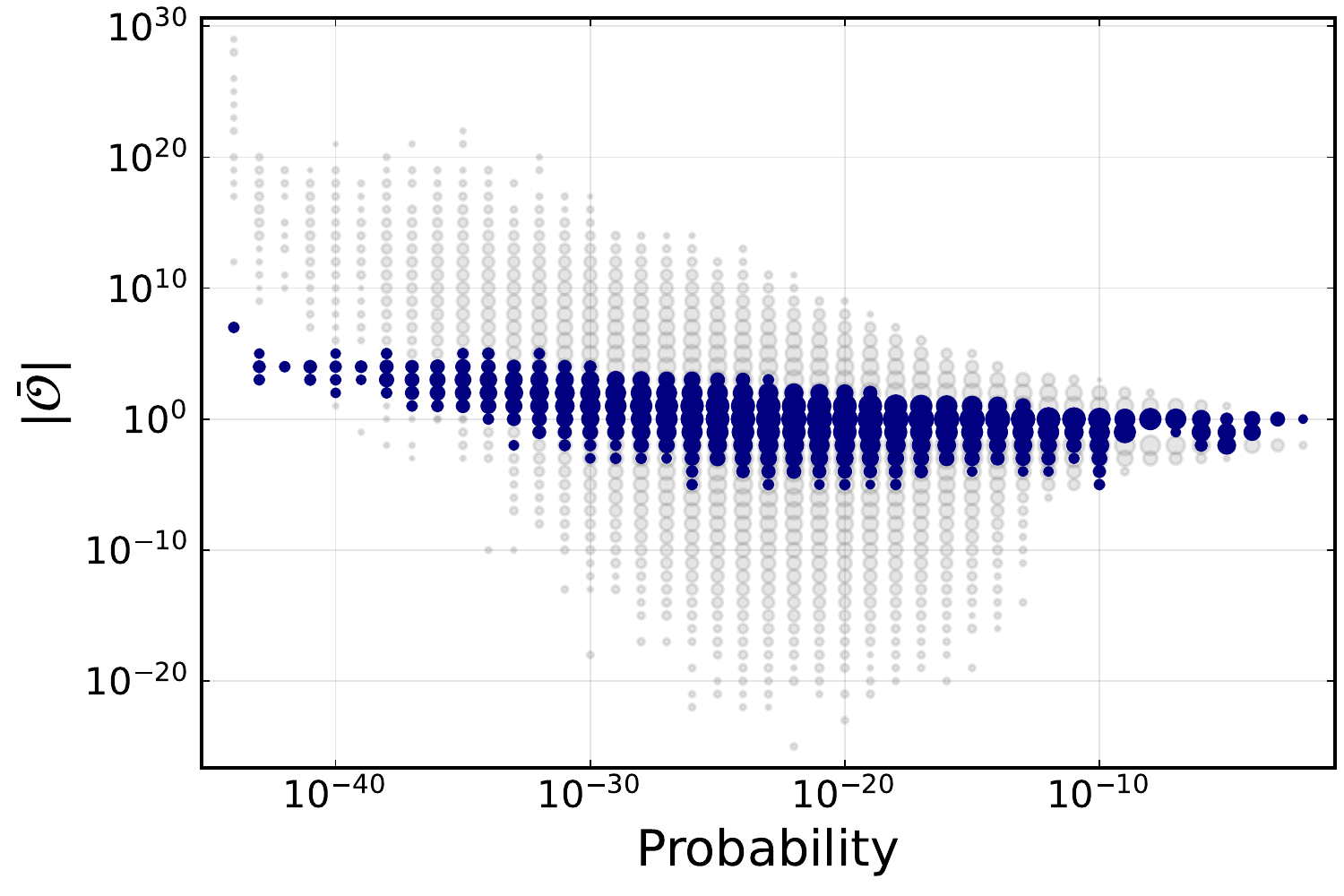}
\caption{The absolute values of the values of condensate vs. the probabilities of the configurations for the $L=4$ lattice with $N_f=2$, $m=-1.5$, $\ss g =2.0$ and the $\he_3$ sampling scheme for the GN model.  Binning is performed by partitioning both axes in intervals of length $1$. The radius of the plot symbol corresponding to a given bin is equal to $1 + \log_{10}($number of samples in the bin$)$. The grey markers are the same as those in Fig.~\ref{fig:6}(c) and are shown here for comparison. 
\label{fig:log-count-condensate}}
\end{figure} 
\begin{figure}[!t]
\includegraphics[width=\columnwidth]{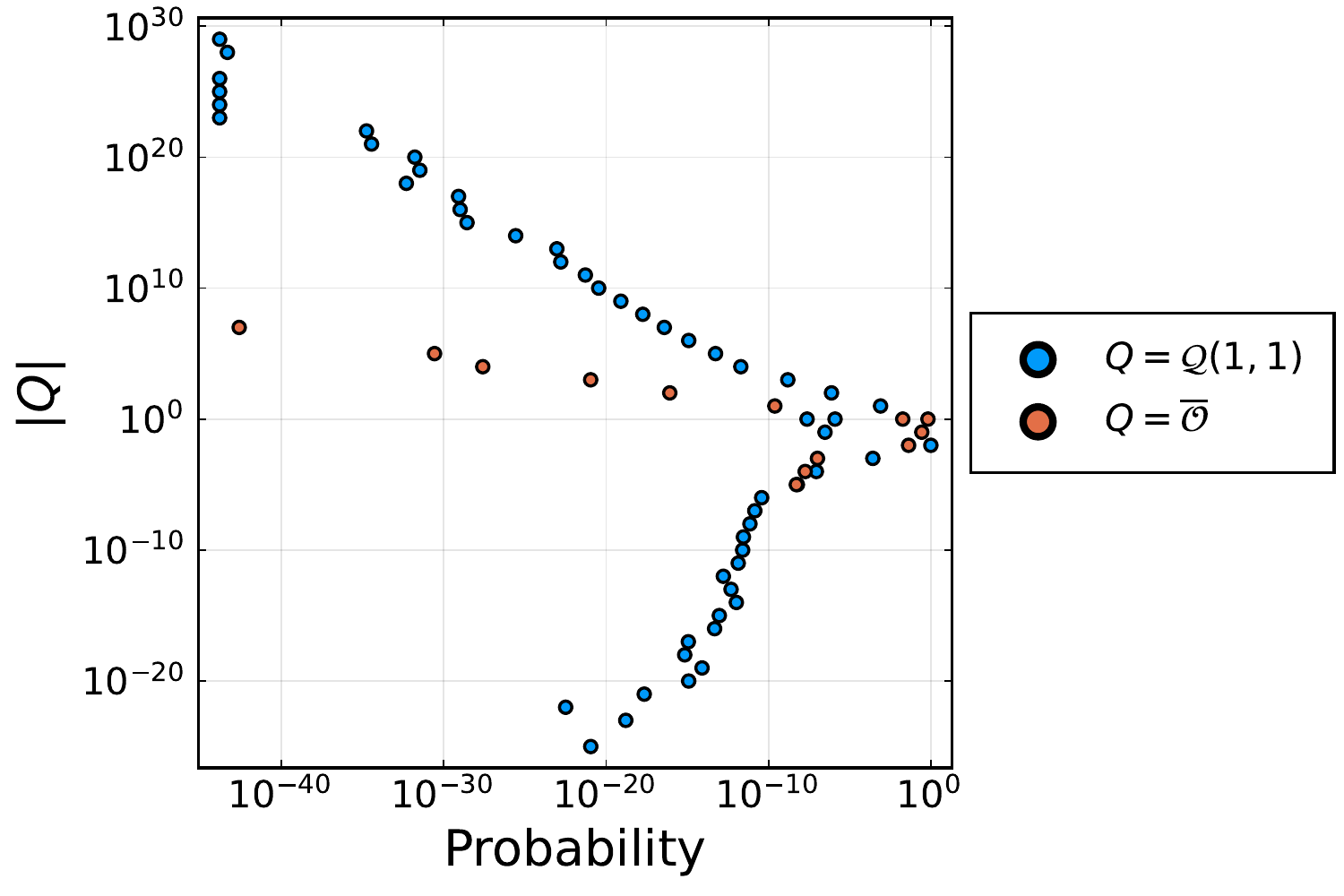}
\caption{Comparison of the spectra of $\cq(1,1)$ and $\overline{\co}$. For each observable, the probabilities of the configurations that correspond to a given observable range (as in Figures \ref{fig:log_count_local_4x4} and \ref{fig:log-count-condensate} respectively) are summed over and provides the probability of finding an observable in the given range.
\label{fig:8}}
\end{figure} 

\subsection{Summary}
\label{sec:discrete-summary}
The discrete sampling schemes that have been proposed have manifestly finite variance provided the roots in the scheme do not contain an exceptional configuration. In the case where the roots do contain an exceptional configuration, the variance is still finite; however, the sample mean will be biased due to the missing  contribution of the exceptional configuration to the mean. These discrete sampling schemes are effective for calculating observables for small lattice volumes and provide interesting testing grounds for investigation of the  fundamental issues of infinite variance. However for quantities with infinite, or very large, variance under a continuous HS sampling, the discrete sampling schemes do not practically overcome the issues of large variance for large volumes.

\section{Reweighting}
\label{sec:reweighting}
In this section, a method for sampling non-negative observables with infinite variance is proposed that constructs the target observable through a series of discrete reweighting steps or through a continuous reweighting procedure. In each case, samplings are performed using probability measures that incorporate part of the observable.

\subsection{Discrete reweighting}

Consider an unnormalised probability distribution $P(x)$ and an observable $\ct(x)$ that is non-negative everywhere. The expectation value
\begin{equation}
  \lc \ct \rc = \ff{\sum_x P(x)\co(x)}{\sum_x P(x)},
\end{equation}
with the standard estimator for this quantity is given by
\begin{equation}
  \hh \ct = \ff{1}{N_s}\sum_{i=1}^{N_s} \ct(x_i),
\end{equation}
where the $x_i$ are sampled with respect to the probability weight $P(x)$. 

As in the previous sections, the  variance of the standard estimator is not well-defined if the second moment of $\ct$ under the unnormalised probability weight $P(x)$ is infinite. 
To surmount this problem,  a set of unnormalised probability weights $P_s(x)$
\begin{equation}
  P_s(x) = P(x)\ct(x)^s,
\end{equation}
are introduced with $P_0(x) = P(x)$. Since $\ct(x)$ is non-negative, this forms a probability distribution for real $s$. We denote the expectation value of an observable with respect to $P_s(x)$  by $\lc \, \cd\,  \rc_s$. It is straightforward to see that:
\begin{equation}
  \lc \ct \rc = \prod_{r=0}^{N-1}\lc \ct(x)^{1/N} \rc_{\ff r N},
  \label{eq:Orw}
\end{equation}
where $N$ is a  positive integer. 

Based on the breakup in Eq.~\eqref{eq:Orw}, an alternative estimator of $\lc \ct \rc$ can be defined as follows.
Consider a set of $N$ configurations such that $x_r$ is sampled with respect to $P_{\ff{r}{N}}$ and denote the set by $\mathbf{x} \equiv \lp x_0,\cdots,x_{N-1} \rp$. Let $\mathbf{x}^{(k)}\equiv \lp x^{(k)}_0,\cdots,x^{(k)}_{N-1} \rp $ for $k\in\{1,\cds, N_s\}$ be i.i.d. sets of configurations. In terms of these sets, a valid estimator is given by
\bad 
\ti \ct\lk \mathbf{x}^{(1)},\cds,\mathbf{x}^{(N_s)} \rk= N_s \sum_{k=1}^{N_s} \prod_{r=0}^{N-1}  \ct^{\ff 1 N} \lp x^{(k)}_r \rp,
\label{eq:reweighting-estimator}
\ead
where the total number of samples is $N_S={N}\times{N_s}$.
It is easy to check that this estimator is unbiased. Except in pathological cases, the random variables $\ct(x_r)^{\ff 1 N}$, where $x_r$ are sampled with respect to $P_{\ff r N}(x)$, will have finite variance for $r\in\{0,\cds,N-1\}$  for large enough $N$ (each quantity is less singular near an exceptional configuration than the original observable). If this is the case, then the estimator $\ti \ct$ will also have finite variance.

Fig.~\ref{fig:gn-2x2} presents results for $\cq(1,1)$  defined in Eq \eqref{eq:gn-op-def} for the Gross-Neveu model on a $L=2$ lattice using this discrete reweighting sampling scheme, Eq.~\eqref{eq:reweighting-estimator}. For comparison with Section \ref{sec:GN-discrete}, Fig.~\ref{fig:gn-2x2} also presents the corresponding results for the $L=8$ lattice. Note that the exact result for the latter case is not shown as calculating it with the $\he_n$ discrete sampling scheme requires the generation of $n^{64}$ configurations. 
\begin{figure*}[!t]
\includegraphics[width=0.95\columnwidth]{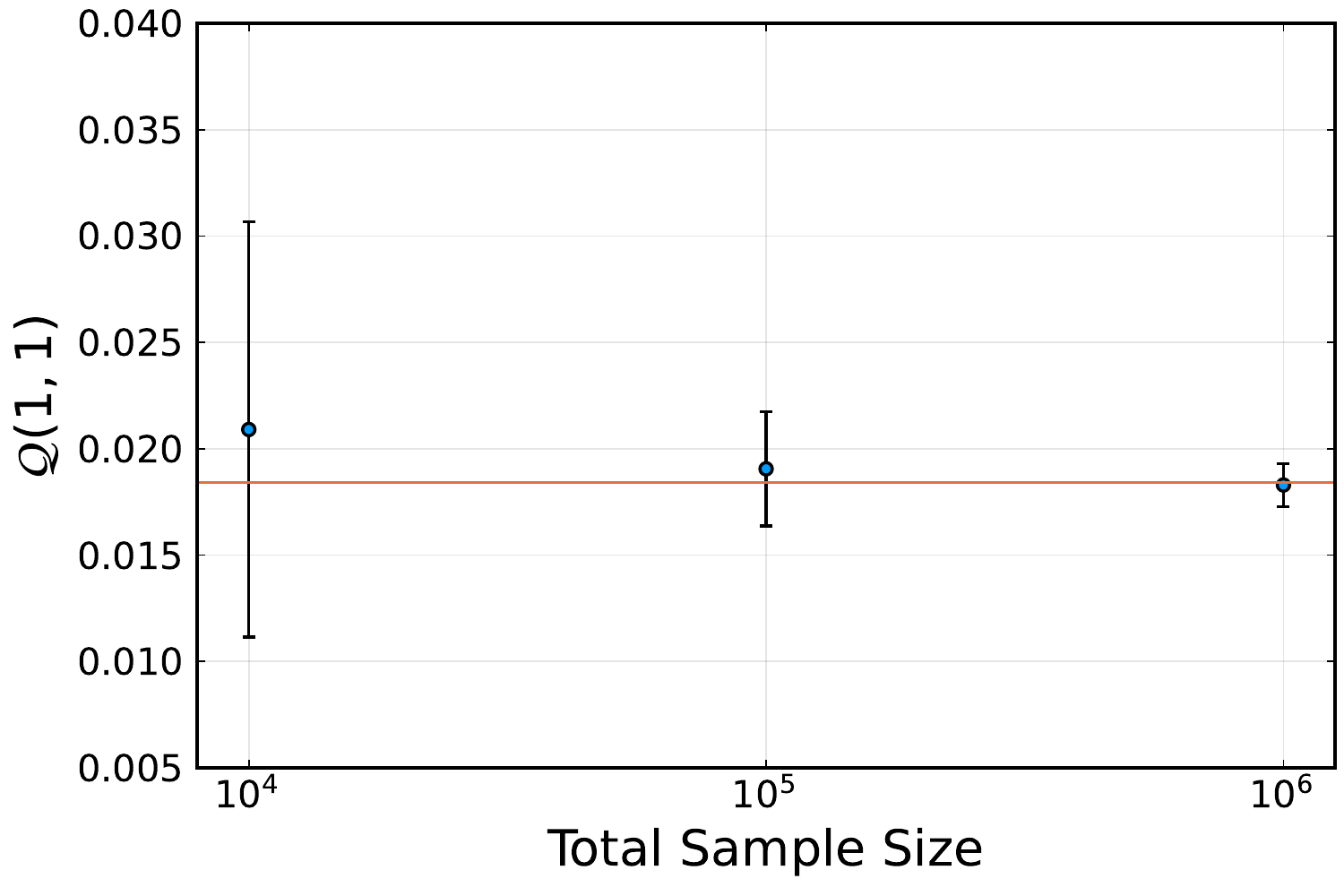}
\qquad \qquad
\includegraphics[width=0.94\columnwidth]{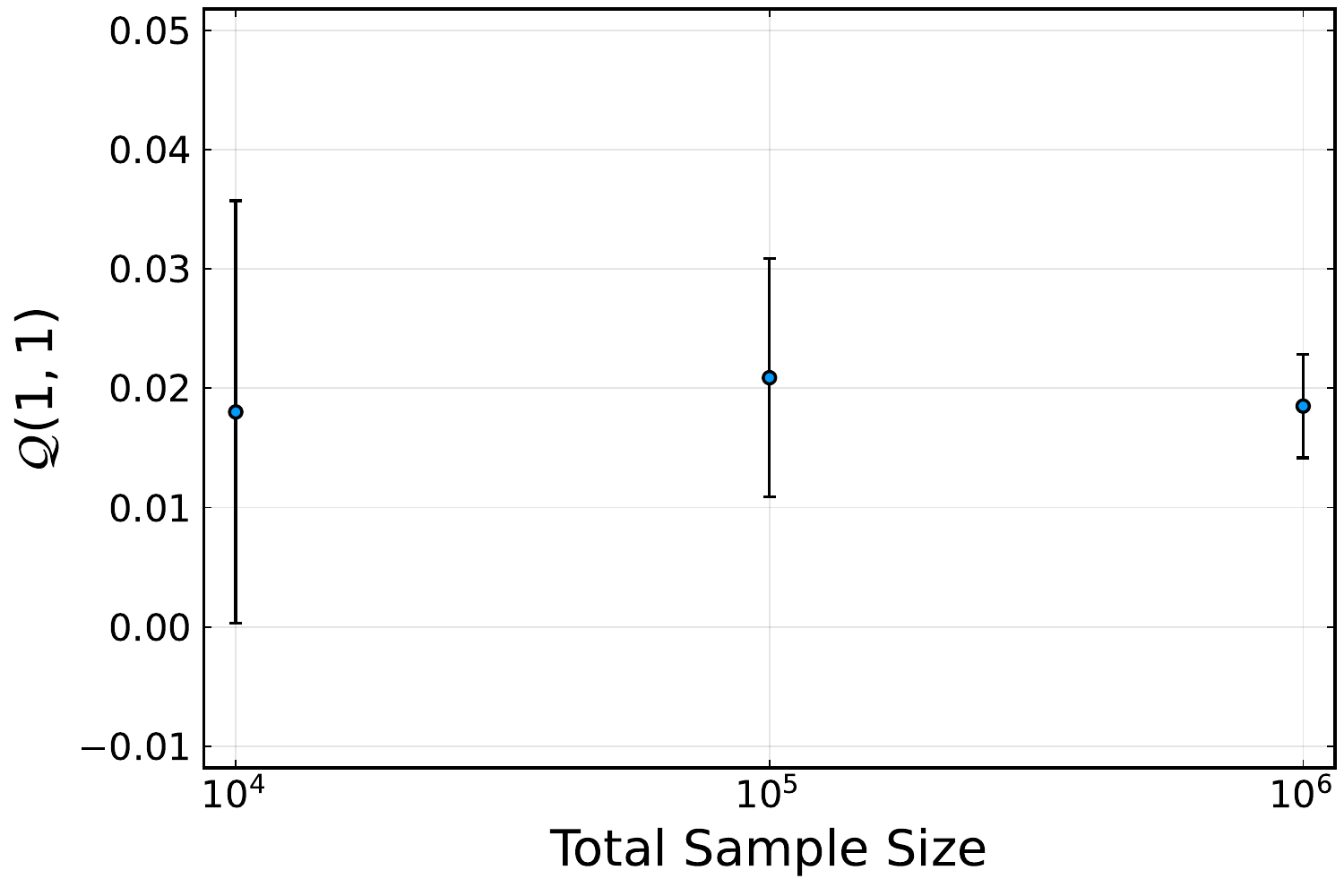}\\
(a)\hspace{8cm}(c)\\[5ex]
\includegraphics[width=0.95\columnwidth]{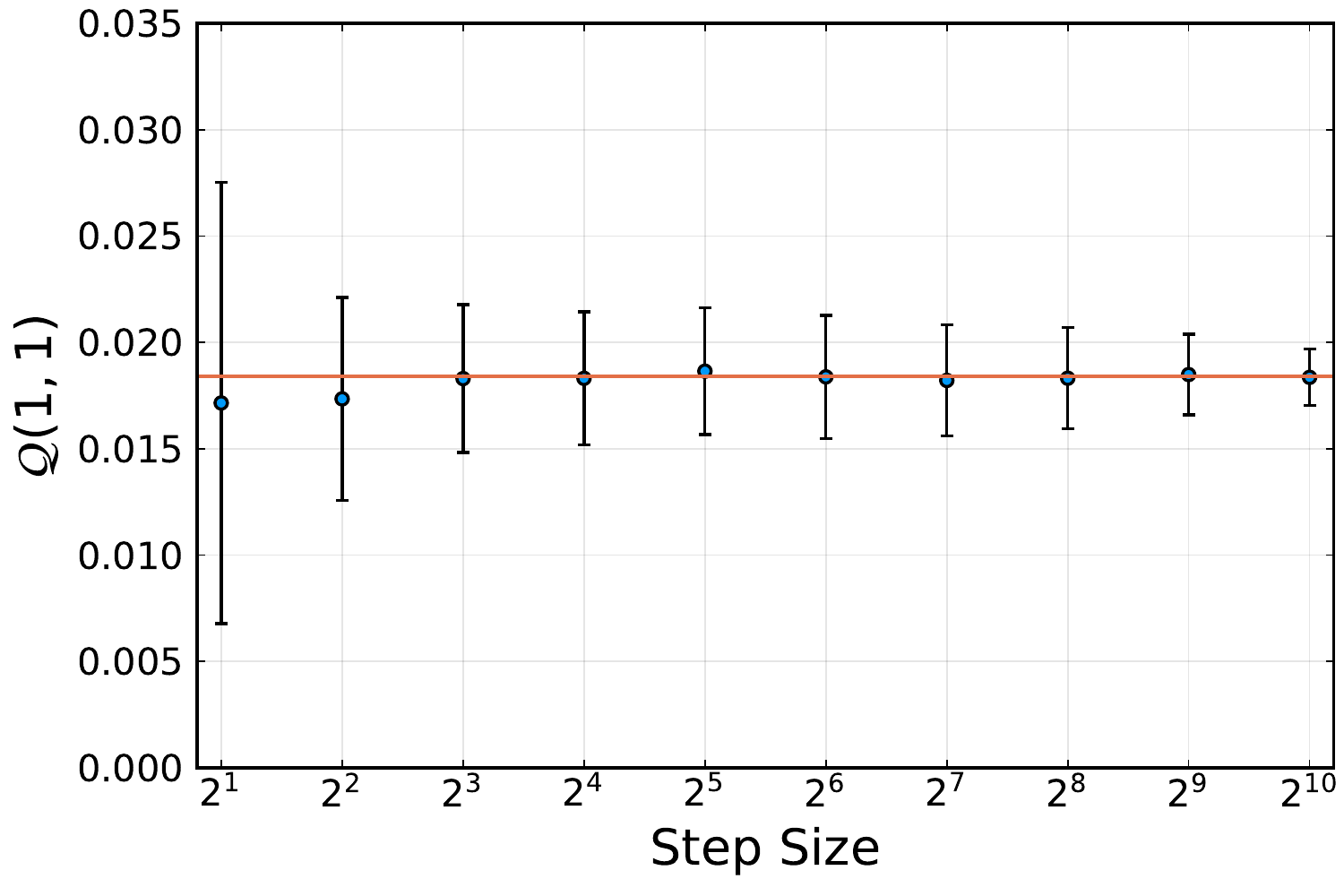}
\qquad\qquad
\includegraphics[width=0.94\columnwidth]{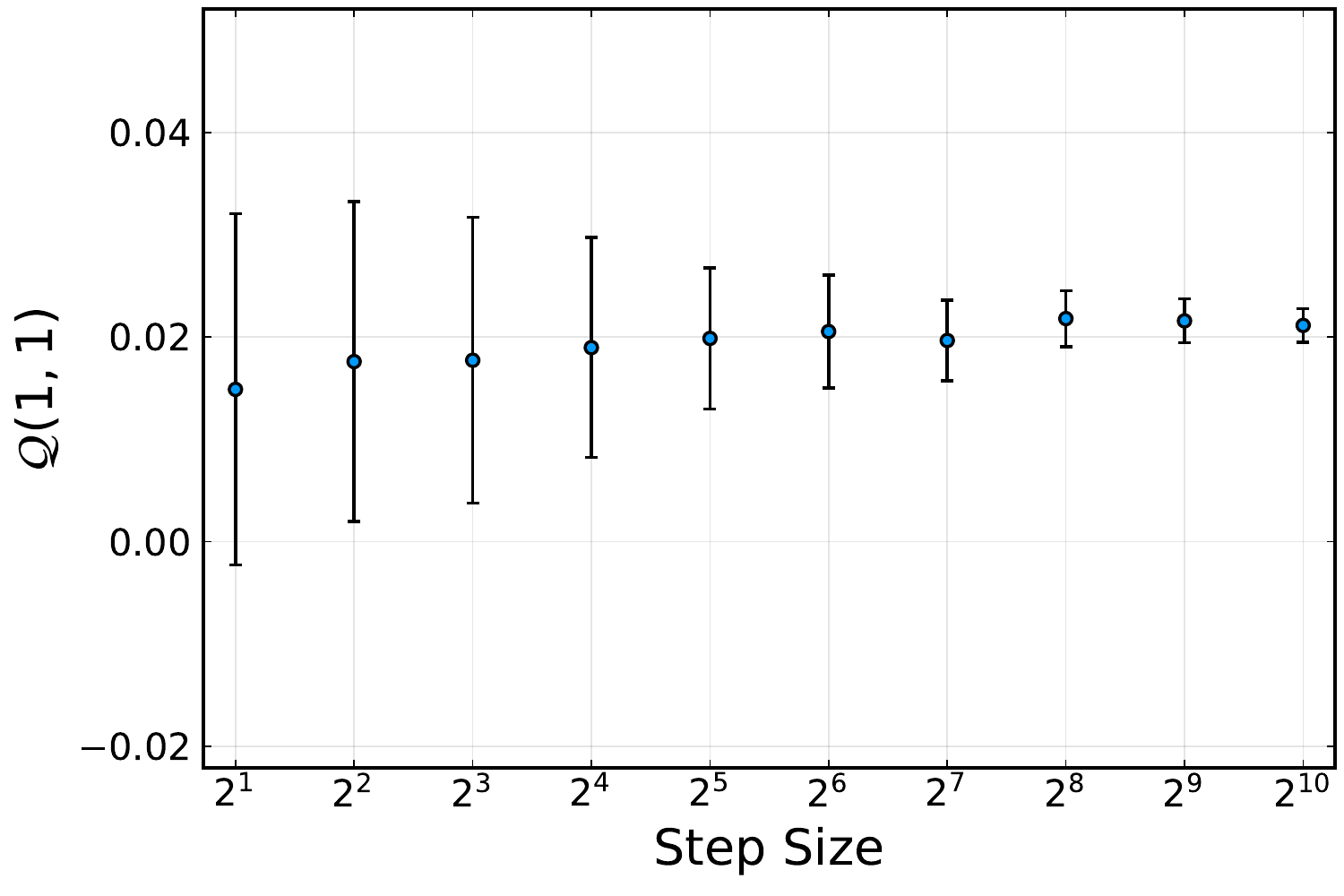}
(b)\hspace{8cm}(d)

\caption{Estimations of the mean of $\cq(1,1)$ obtained with the median of means estimator applied to Eq.~\eqref{eq:reweighting-estimator} for various step numbers and  total sample sizes for the Gross-Neveu model for the $L=2$ (left) and $L=8$ (right) lattice extent, and for $m=-1.5$, and $\ss g =2.0$. In the top row, (a) and (c), the step size $N=10$ is fixed while in the bottom row, (b) and (d), the total sample size $N_T=1048576=2^{20}$ is fixed and step sizes are chosen to be $2^k$ for $k \in\{1,\cds,10\}$. For $L=2$, the red line shows the exact value obtained from explicit summation over all possible configurations of the discrete sampling scheme. The error bars at each sample size show a  confidence level of $0.9973$. 
}
\label{fig:gn-2x2}
\end{figure*}

\subsection{Continuous Reweighting}

There is a natural extension of this sequential reweighting method  to a continuous version of the  procedure. To arrive at this version, note that
\begin{equation}
\lc \cp \rc = \lc \cp^{1-s} \rc_s \lc \cp^s \rc_0 
\end{equation}
for any $s \in [0,1]$ and positive observable $\cp$ as long as $Z_s \equiv \lc \cp^s \rc_0$ is finite. $Z_s$ is naturally interpreted as the partition function for the probability weight $P_s(x) = P(x)\cp(x)^s$. Since the left-hand side of this equation is $s$-independent, one obtains 
\begin{equation}
0 = \ff{d}{ds}\lp \lc \cp^{1-s}\rc_s Z_s\rp.
\end{equation}
This differential equation is straightforward to solve, noting that $\inv{Z_s}Z'_s = \lc \log \cp \rc_s$ and $\lc\log(\cp^0)\rc_s=0$. Therefore, under the assumption that $Z_s$ and $\lc \log \cp\rc_s$ are both finite for $s \in [0,1]$, one finds that
\begin{equation}
\lc \cp \rc ={\rm exp}\lp\int_{0}^1 ds \lc \log \cp \rc_s \rp.
\label{eq:log-eq}
\end{equation}

Utilizing Gauss-Legendre quadrature and  Eq.~\eqref{eq:log-eq}, an estimator for $\log \lc \cp \rc$ can be defined. Let $N$ a positive integer. Then an integral of the form $\int_0^1 ds\, f(s)$ can be approximated by $\sum_{i=1}^N c_i f(s_i)$, where $c_i = \ff 1 2 w_i$ and $s_i=\ff{1+z_i}{2}$ are determined by the roots, $z_i$, of the $N$th Legendre polynomial and the corresponding weights, $w_i$, associated with Gauss-Legendre quadrature.\footnote{If $f(x)$ is a polynomial of degree at most $2N-1$, then $\int_{-1}^1 f(x)dx = \sum_{i=N}w_i f(x_i)$ where $\{x_i\}$ are the roots of the $N$th Legendre polynomial $P_N(x)$ and $w_i = \ff{2}{(1-x_i)^2 \lp P'_N(x_i)\rp^2}$.} Defining the set of configurations $\mathbf{x} = \lp x_1,\cds, x_N\rp$  where $x_i$ is sampled with respect to $P_{s_i}(x)$, let $\mathbf x^{(k)}$ for $ k \in\{1,\cds,N_s\}$ be i.i.d. sets of configurations. Then the following is an estimator for $\log \lc \cp \rc$:
\bad 
\widetilde{\log \lc \cp \rc} \lk \mathbf{x}^{(1)},\cds,\mathbf{x}^{(N_s)} \rk = N_s\sum_{k=1}^{N_s}\sum_{i=1}^{N} \log \cp \lp x_{i}^{(k)}\rp,
\label{eq:log-estimator}
\ead 
where $N_S=N\times N_s$ is the total number of samples divided evenly between each of the $N$ Gauss-Legendre quadrature points.

As a very simple demonstration of this method, consider a one-dimensional example of the standard normal  distribution $P(x)=\ff{1}{\ss{2\pi}}e^{-x^2/2}$ and observables $\cp_p(x) = e^{px^2}$ for real $p < \ff 1 2$. The expectation values of these observables are given by $\lc\cp_p\rc=\ff{1}{\ss{1-2p}}$ and the variances by ${\rm var}(\cp_p) = \ff{1-2p-\ss{1-4p}}{(1-2p)\ss{1-4p}}$. The variance diverges for $p \geq \ff 1 4$ and the expectation value diverges for $p\geq\ff 1 2$
\begin{figure}[!ht]
	\includegraphics[width=\columnwidth]{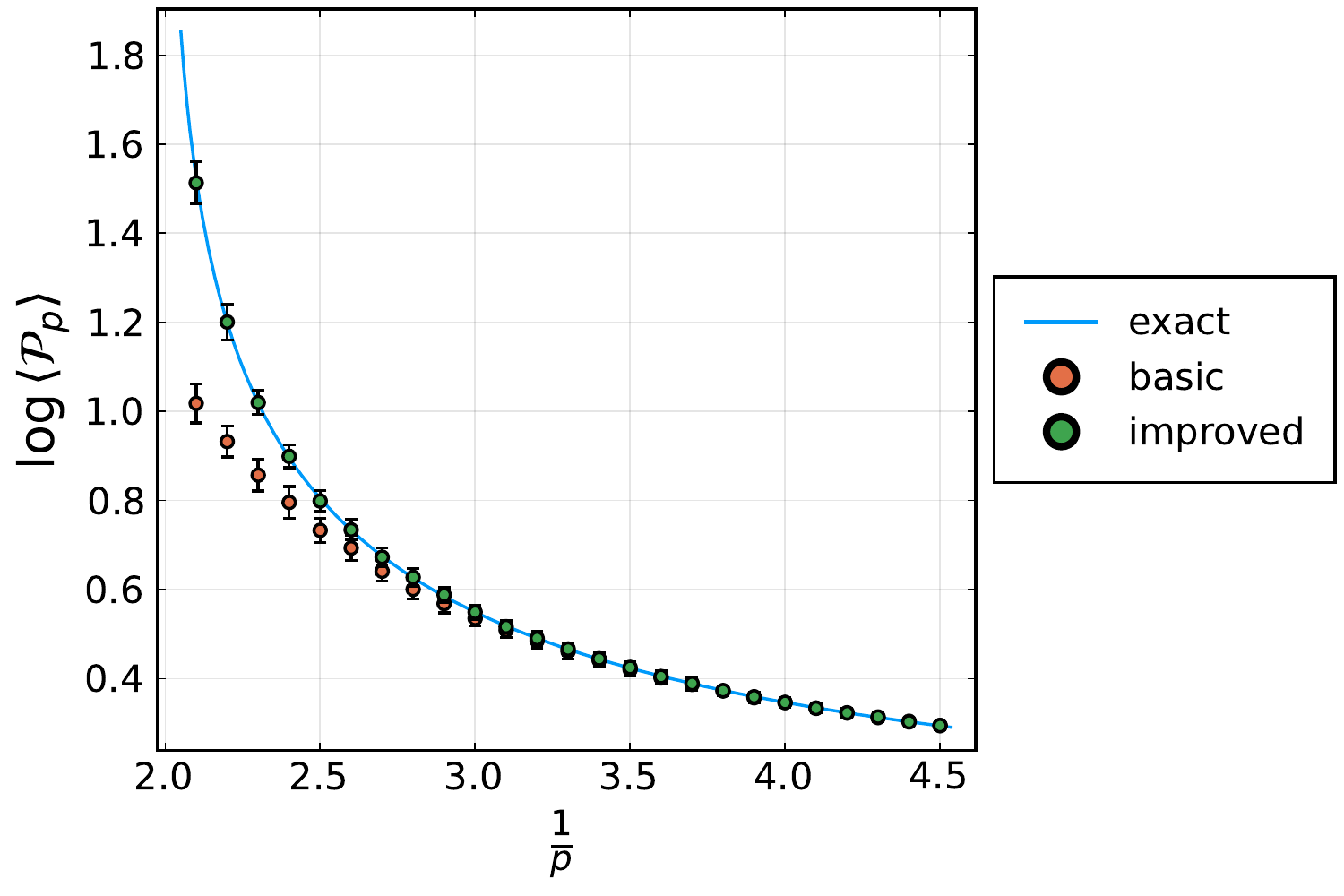}
	\caption{
	Estimates of $\log \lp \cp_p \rp$ obtained with the median of the means estimator applied to Eq.~\eqref{eq:log-estimator} as a function of $1/p$ for $N=1000$ nodes and $N_s = 1000$ samples per node in the Gauss-Legendre sampling. The blue line shows the exact values. ``Basic'' corresponds to  the results obtained where the random variables for each distribution $P_{s \neq 0}$ are obtained through the Metropolis algorithm with $P_0$ being the proposal distribution. ``Improved'' corresponds to the results where the random variables are drawn from the distribution $P_s$  directly.  Error bars show a statistical confidence level of $cl = 0.9973$. 
	\label{fig:log-mom}}
\end{figure}

To calculate $\log \lc \cp_p \rc$ using Eq.~\eqref{eq:log-estimator} and Gauss-Legendre quadrature, a proposal distribution for each $s_i$ is required. A simple and straightforward choice is to use $P_0(x) = \ff{1}{\ss{2\pi}} e^{-\ff 1 2 x^2}$ for all $s_i$. The data labelled as ``basic'' in Fig.~\ref{fig:log-mom} is generated in this context. For $1/p\agt 3$ this provides an accurate estimate that agrees with the exact value within uncertainties. However for a fixed sample size per Gauss-Legendre node, the estimates deviate from the exact value as $p$ approaches $\ff 1 2$.
A possible cause of this is that for $p = \ff 1 2-\e$, the probability weight is proportional to $e^{-\e x^2}$ for $s = 1$ and important contributions will be due to $\abs{x} \lesssim \ff{1}{\ss \e}$. Therefore, more and more samples will be needed as $\e\to 0$ and the ``basic'' method suffers from an overlap problem. To improve the algorithm in this simple example, one can also directly sample for $P_{s_i}$ which are normal distributions. Results generated in this latter context are labelled as ``improved'' in Fig.~\ref{fig:log-mom}, and are seen to agree perfectly with the exact results for all $p$.
Systematic errors due to the finite number of nodes are negligible compared to the statistical errors in both cases.

To test this method on a more realistic system, estimates of $\widetilde{\log \cq(1,1)}$ for the Gross-Neveu model on the $L=2$ and $L=8$ lattices with $m=-1.5$, $\ss g = 2.0$ and $N_f = 2$ are presented in Fig.~\ref{fig:gn-log-2x2}. For $L=2$, the exact value is reproduced within uncertainties while for $L=8$, results from continuous reweighting agree with those from discrete reweighting.
\begin{figure*}[!th]
\includegraphics[width=0.95\columnwidth]{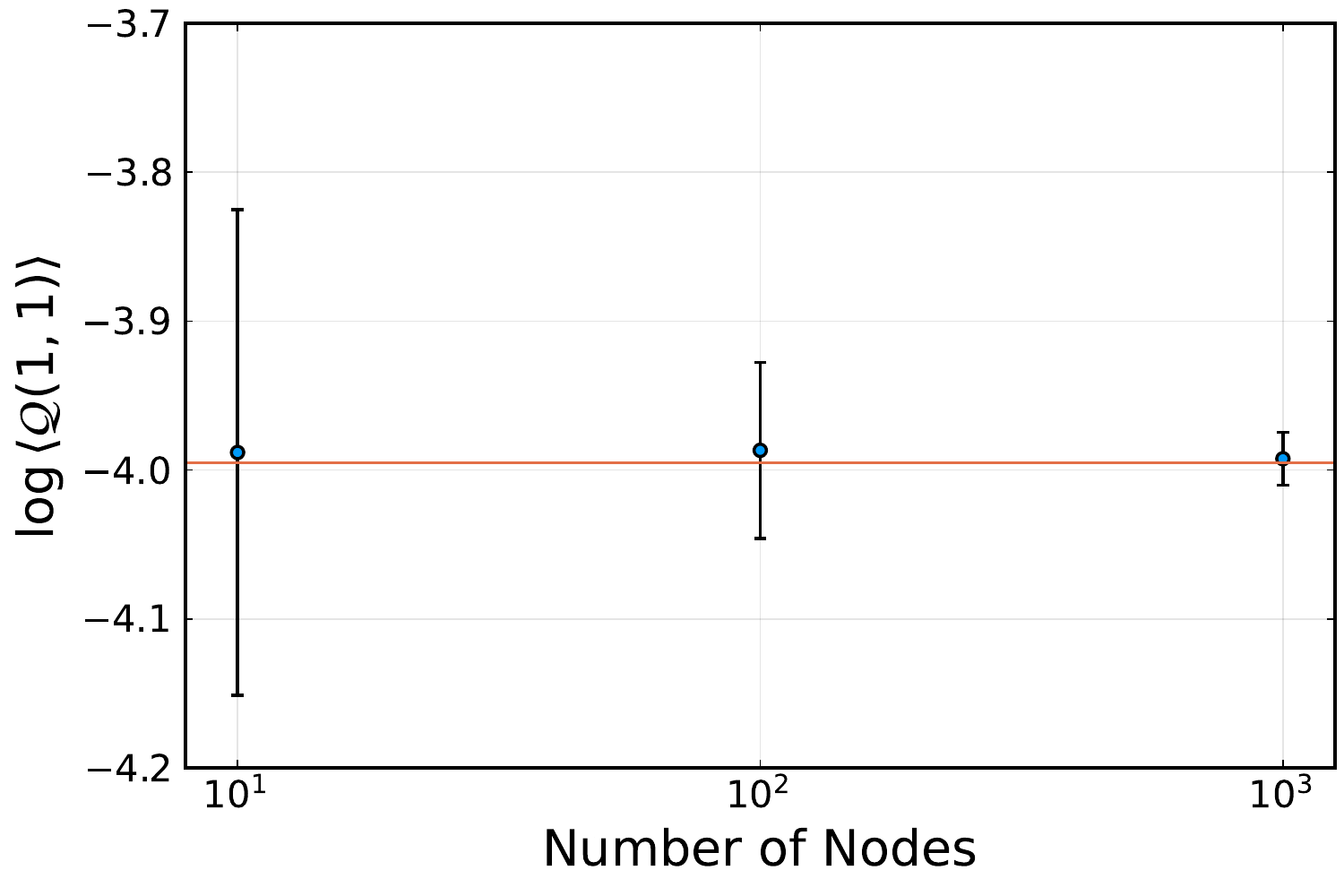} \qquad \qquad
\includegraphics[width=0.94\columnwidth]{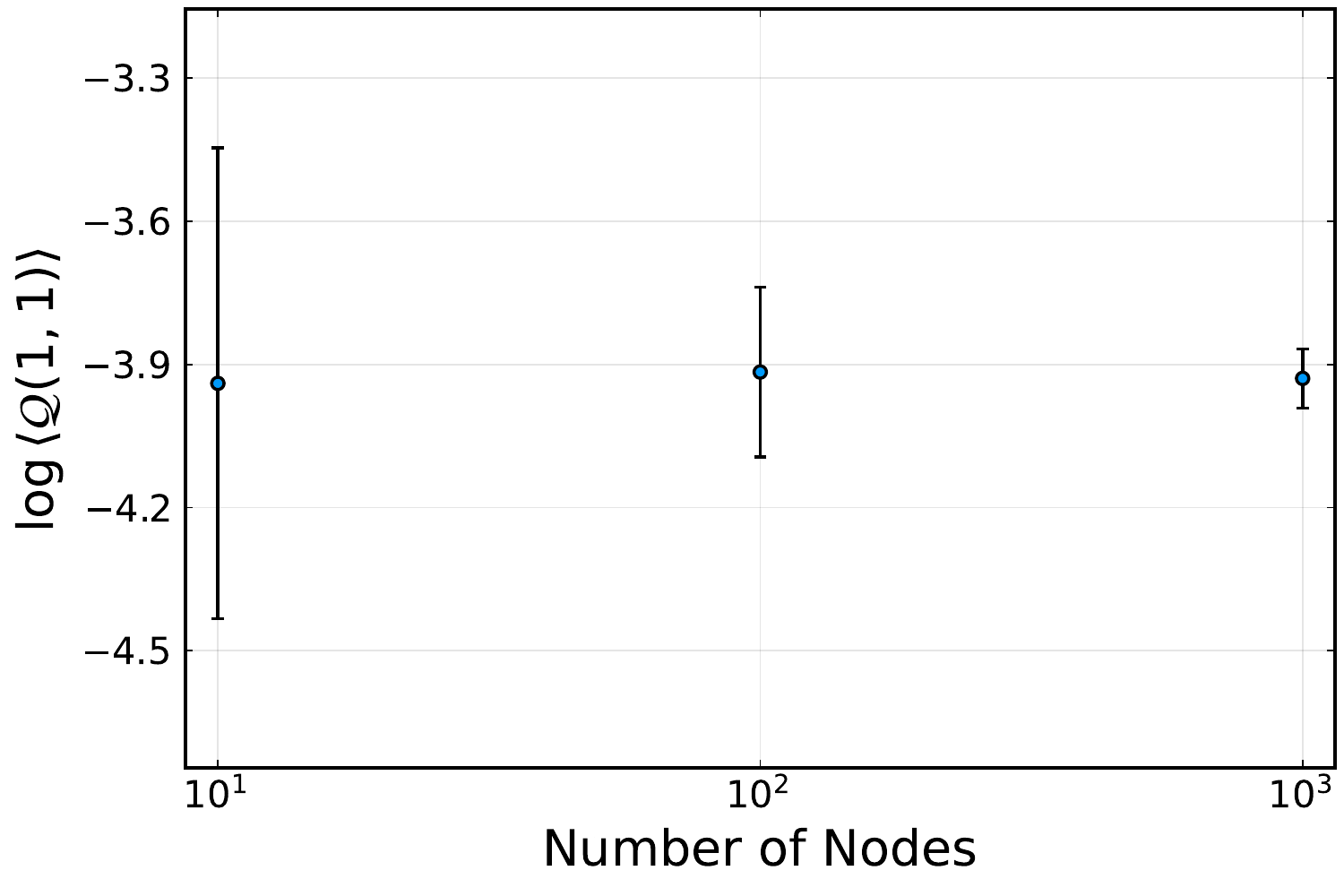} \\
(a) \hspace{8cm}(c) \\[5ex]
\includegraphics[width=0.95\columnwidth]{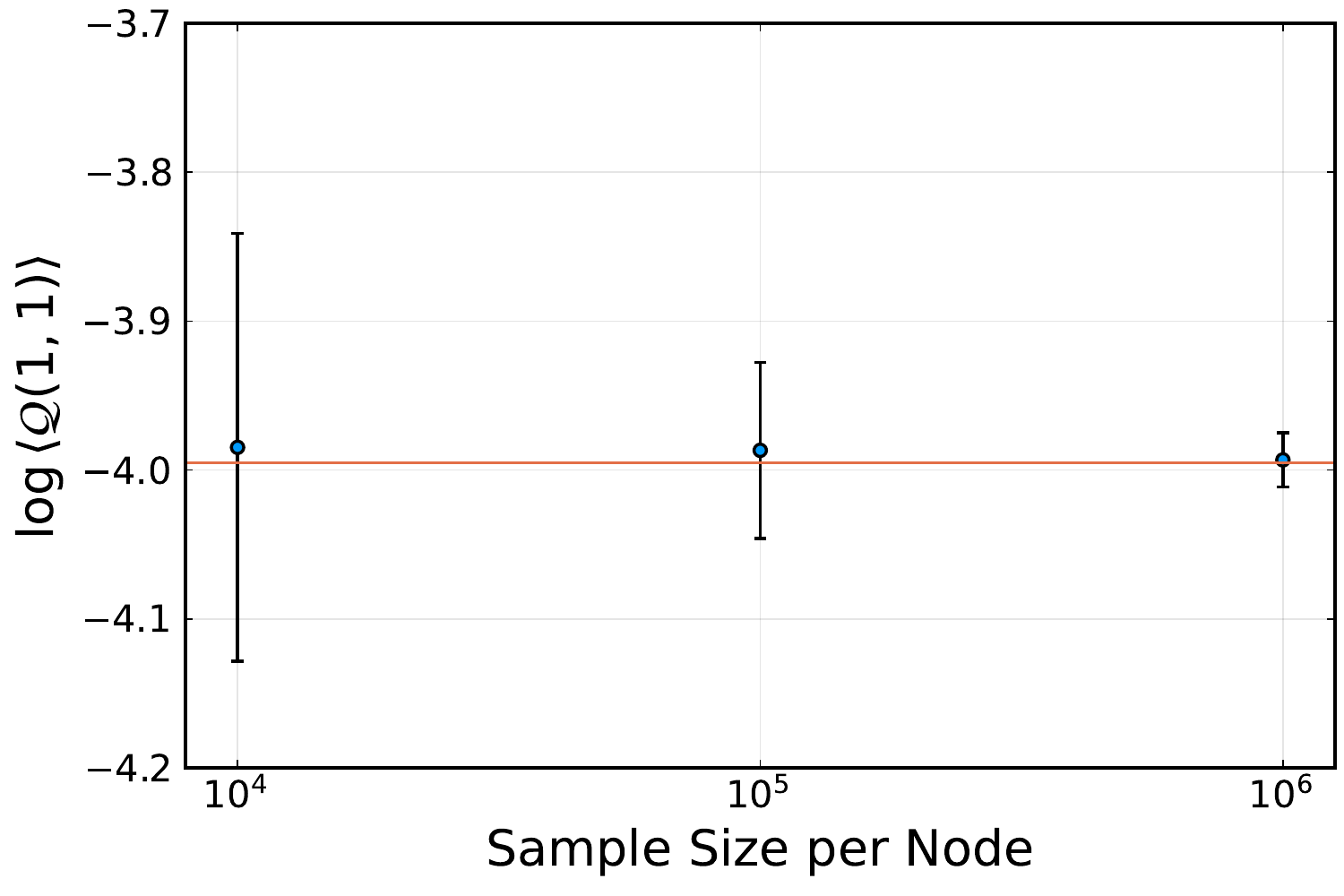}
\qquad \qquad
\includegraphics[width=0.94\columnwidth]{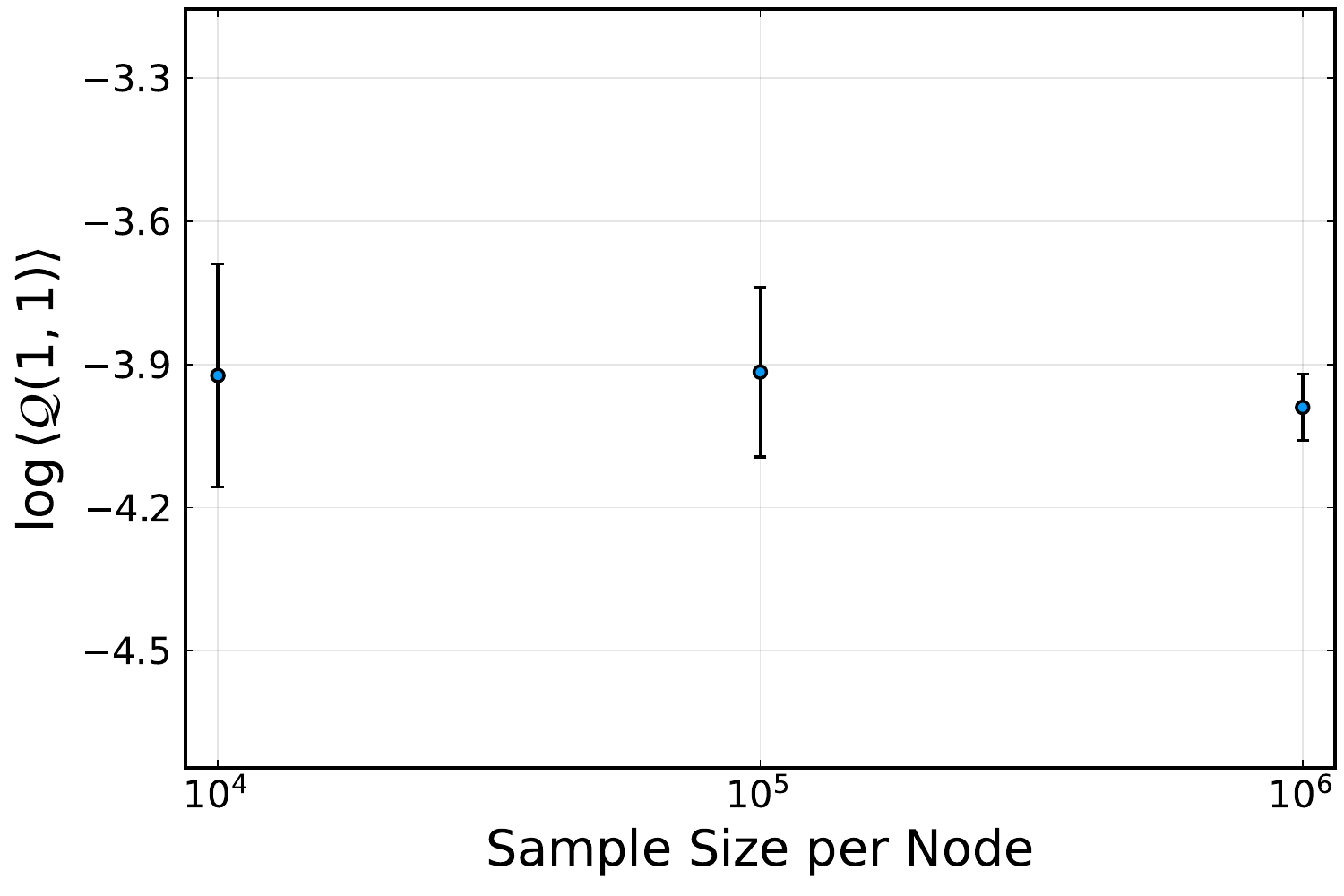} \\
(b) \hspace{8cm}(d)

\caption{Estimates of $\log \lc \cq(1,1) \rc $, for the Gross-Neveu model on  $L=2$ (left column) and $L=8$ (right column) lattice extents with $m=-1.5$, $\sqrt{g} =2.0$ and $N_f=2$, obtained with the median of means estimator applied to Eq.~\eqref{eq:log-estimator} for $N=10^1,10^2,10^3$ nodes with sample size per node $N_{SSPN} = 10^5$ fixed in the top row, (a) and (c), and for $N_{SSPN} = 10^4, 10^5, 10^6$ sample sizes per nodes with the number of nodes $N=100$ fixed in the bottom row, (b) and (d). Error bars show a confidence level of $cl = 0.9973$. In the $L=2$ case, the red line shows the exact value obtained from explicit summation over all possible configurations of the discrete sampling scheme. 
}
\label{fig:gn-log-2x2}
\end{figure*}

\section{Conclusions}
\label{sec:summary}
Large statistical variance in Monte Carlo sampling  severely limits the precision with which many important quantities in quantum field theories can be  determined. In this work, quantities that have formally infinite variance under standard sampling schemes have been considered. 
In the context of fermionic theories, a family of discrete sampling schemes has been presented that surmounts the issue of infinite variance. Nevertheless, the variances in these schemes can be very large (compared to their means squared) and hence sampling maybe inefficient. An alternate sampling scheme has also been developed which can be applied to any non-negative random variable that can be sampled with the Monte Carlo method. While the method has been proposed in order to estimate observables with infinite variances, it is likely to be effective for non-negative random variables that have finite but large noise to signal ratios. There are potentially interesting connections of the investigations presented here to the large time-separation behaviour of two-point correlation functions for quantities that possess global charges, such as for baryons and nuclei in QCD, that will be explored in subsequent work.

\acknowledgements{
We are grateful for insightful discussions with P. Shanahan.
This work is supported by the National Science Foundation under Cooperative Agreement PHY-2019786 (The NSF AI Institute for Artificial Intelligence and Fundamental Interactions, http://iaifi.org/). WD is supported in part by the U.S.~Department of Energy, Office of Science, Office of Nuclear Physics under grant Contract Number DE-SC0011090. WD is also supported by the SciDAC4 award DE-SC0018121. This research was supported in part by the National Science Foundation under Grant No. NSF PHY-1748958.
} 

\appendix
\section{Review of the Basic Probability Theory} 
\label{appendix:probability}
The standard method of estimating an observable (a random variable) is based on the Central Limit Theorem. We therefore begin by reviewing the background for the Central Limit Theorem. We refer to Ref.~\cite{durrett_2019} for further details.

A probability space is a triplet ($\O,\calf,P$) where:
\begin{itemize}
    \item $\O$ is the sample space. 
    \item $\calf$ is the space of events and is required to be a $\s$-algebra.
    \item $P:\calf \to [0,1]$ is the probability measure. 
\end{itemize}
Every element $\o$ in the sample space $\O$ is called an outcome. An event $A \in \calf$ is said to occur if $\o$ is the outcome and $\o \in A$. $\calf$ is required to be a $\s$-algebra which means that
\begin{itemize}
    \item $\emptyset, \O \in \calf$.
    \item $A \in \calf$ implies $A^\mathsf{c} \in \calf$ where $A^\mathsf{c}$ is the complement of $A$.
    \item If $A_{i\geq 1} \in \calf$ is a countable sequence of elements of $\calf$ it follows that $\cup_{i\geq 1} A_i \in \calf$.
\end{itemize}
The probability measure $P$ is required to satisfy:
\begin{itemize}
    \item $P(\emptyset) = 0$.
    \item $P(\O) = 1$. 
    \item If $A_{i\geq 1} \in \calf$ is a countable sequence of pairwise disjoint elements of $\calf$ it follows that $P\lp \cup_{i\geq 1} A_i \rp = \sum_{i\geq 1} P(A_i)$.
\end{itemize}

It must be noted that in general one can't choose $\calf = 2^\O$, the set of all subsets of $\O$. Therefore, the choice of $\calf$ is essential and elements of $\calf$ are said to be measurable.

A random variable $X:\O \to \mathbb{R}$ is a real valued function on the sample space such that for all $a \in \mathbb{R}$, $\inv X \lp A_{\calb_{\mathbb{R}}} \rp \in \calf$ where $\calb_\mathbb{R}$ is the Borel $\s$-algebra on $\mathbb{R}$, the smallest\footnote{Smallest $\s$-algebra containing a given set of sets is defined as the intersection of all $\s$-algebras that contains the given set of sets which can be shown to be a $\s$-algebra.} $\s$-algebra containing all open subsets of $\mathbb{R}$. An equivalent condition is $\inv X\lp (-\ii,a]\rp \in \calf$ for all $a \in \mathbb{R}$. This condition allows one to define another probability distribution $P_X$ on the real line through the formula $P_X \lp (a,b] \rp = P \lp \inv X \lp (-\ii,b] \rp \rp - P \lp \inv X \lp (-\ii,a] \rp \rp$. By the celebrated Carathéodory's extension theorem $P_X$ can be extended to the $\calb_\mathbb{R}$. If $A \in \calb_\mathbb{R}$, $P_X(A)$ should be interpreted as the probability that $X$ takes value in $A$. We further define the cumulative distribution function $F_X(t) = P_X \lp (-\ii,t] \rp$ which gives the probability that $X \leq t$. 

For a stochastic physical system represented by the probability space $\lp \O, \calf, P\rp$, one can consider the space of repeated outcomes denoted by $\O^\ii = \prod_{i=1}^\ii \O$, associated with the $\s$-algebra $\calf^\ii$, the smallest $\s$-algebra containing $\prod_{i=1}^\ii A_i$ where only finitely many of $A_i \in \calf$ are different than $\O$. An element $\o^\ii \in \O^\ii$ is given by $\o^\ii = \{\o_1, \o_2, \cdots \}$ where $\o_i \in \O$ for all $i \in \mathbb{N}^+$. Then, $\O^\ii$ can be identified with the set of of all samples of the physical system with infinite sample size. By Kolmogorov's Extension Theorem, there exists a unique probability measure $P^\ii$ on $\calf^\ii$ such that $P^\ii \lp \prod_{i=1}^\ii A_i \rp = \prod_i P(A_i)$ if only finitely many of $A_i \in \calf$ are different that $\O$. Given a random variable $X$ on $(\O,\calf,P)$, we define the random variable $X_n$ on $(\O^\ii,\calf^\ii,P^\ii)$ by $X_n(\o^\ii) = X(\o_n)$.

For our purposes, there are three important types of convergence for random variables. One says that a sequence of random variables $X_n$ converges to a random variable $X$ 
\begin{itemize}
    \item almost surely/everywhere\footnote{To be precise, there is a set $A \in \calf$ such that for all $\o \in A$, $\lim_{n\to \ii} X_n(\o) = X(\o)$ and $P(A) = 1$. It is possible that the set of all elements $\o \in \O$ satisfying $X_n(\o) \to X(\o)$ is not measurable, but this distinction is not relevant for our discussion. \label{ftnt:almost-surely}} if $P(X_n \to X) = 1$. It is denoted by $X_n \overset{a.s.}{\longrightarrow} X$. 
    
    \item in probability if for every $\e > 0$, $P\lp \abs{X_n-X} > \e \rp \to 0$. It is denoted by $X_n \cip X$.
    
    \item in distribution if $F_{X_n}$ converges pointwise to $F_X$ at every continuity point of $F_X$. It is denoted by $X_n \cid X$.
\end{itemize}
These three different types of convergence imply each other in the sense that
\be 
    \cas \ \Longrightarrow \ \cip \ \Longrightarrow \ \cid .
\ee

Three of the most important results of the probability theory are the Strong Law of the Large Numbers (SLLN), The Weak Law of Large Numbers (WLLN), and the Central Limit Theorem (CLT). 

\ul{\bf Strong Law of Large Numbers}
Let $\{X_n\}$ be a sequence of identically and independently distributed random variables with the finite mean $\mathbb{E}[X_n]=\m$. If one defines the sample mean $\bb X_{n \geq 1} = \frac{1}{n} \sum_{i=1}^n X_i$, it follows that: 
\be 
\bb X_n \cas \m 
\label{thm:SLLN}
\ee 

An important consequence of the SLLN is the Weak Law of Large Numbers

\ul{\bf Weak Law of Large Numbers}
Let $\{X_n \}$ be a sequence of identically and independently distributed random variables with the finite mean $\mathbb{E}[X_n]=\m$. If one defines the sample mean $\bb X_{n} = \frac{1}{n} \sum_{i=1}^n X_i$ for $n\ge 1$, it follows that for any $\e>0$: 
\be 
\lim_{N\to \ii} P\lp\abs{\bb X_N - \m}\leq \e\rp = 1
 \label{thm:WLLN}
 \ee

Although the SLLN  says any sequence of sample means will eventually converge to the mean, in practice it does not say anything about how close a sample mean is to the mean for a given sample size $N$. A similar statement also applies to the WLLN. 
On the other hand, the CLT gives a measure of how close the sample mean is to the mean.

\ul{\bf Central Limit Theorem}
Let $\{X_n\}$ be a sequence of identically and independently distributed random variables with the finite mean $\mathbb{E}[X_1]=\m$ and finite variance $\mathbb{E}[X^2_1]-\mathbb{E}[X_1]^2 = \s^2$. It then follows that $\sqrt{n}\lp \bb X_n - \m \rp \cid \caln(0,\s^2)$ where $\caln(\m,\s^2)$ denotes the normal distribution with  mean $\m$ and  variance $\s^2$.

It follows that for large enough $n$, $\ss{n}\lp \bb X_n - \m \rp \approx \caln(0,\s^2)$. Therefore, it follows that $\bb X_n \approx \caln(\m,\ff{\s^2}{n})$. Although one can derive a similar expression for the estimation of $\s^2$ in the case $\mathbb{E}[X_1^4] < \ii$, it is usually enough to estimate $\s^2$ by the (unbiased) estimator $s_n = \frac{1}{n-1}\sum_{i=1}^n \lp X_i - \bb X_n \rp^2$. 

Let $\{X_n\}$ to be an i.i.d. sequence of random variables with finite variance $\s^2$. Application of the SLLN to $\{X_n\}$ and $\{X^2_n\}$ immediately implies that $s_n \cas \s^2$.

\subsection*{Theorems under infinite variance} \label{sec:main-theorems}

Two theorems are particularly important for analysis of random variables with infinite variance.

\begin{thm} Let $X_{n\geq 1}$ to be a sequence of independent and identically distributed random variables with finite mean $\m$ and infinite variance. Then, for any given $L > 0$, the number of the random variables $s_n$ that satisfies $s_n > L$ is infinite almost surely. \thlabel{thm:infinite-jump}
\end{thm}

We need a bit preparation before we can prove \thref{thm:infinite-jump}.

\begin{lem}[Second Borel-Cantelli]
Let $\{E_n\}$ to a sequence of independent events. If $\sum_{n} P(E_n) = \ii$ then\footnote{Here, $\limsup E_n \equiv \cap_{n\geq 1} \cup_{m \geq n} E_m$ and is equal to the set of all outcomes $\o$ such that $\o \in E_k$ for infinitely many $E_k$. Therefore, $P(\limsup E_n)$ can be interpreted as the probability that infinitely many events $E_k$ happens.} $P(\limsup E_n) = 1$. 
\end{lem}

\begin{lem}
Let $Z$ be a non-negative random variable with infinite mean. Then, for any given $L > 0$, $\sum_{n=1}^\ii P(\{Z \geq  nL\}) = \ii$. \thlabel{lem:probsum}
\end{lem}
\begin{proof}
\bad 
\ii &= \mathbb{E}[Z] \\ 
    &= \sum_{n=0}^\ii \int_{nL\leq Z <(n+1)L} ZdP \\ 
    &\leq L\sum_{n=0}^\ii (n+1) P(\{ nL \leq Z < (n+1)L\}) \\
    &\leq L\sum_{n=0}^\ii n P(\{ nL \leq Z < (n+1)L\}) \\
    &\quad + L\sum_{n=0}^\ii P(\{ nL \leq Z < (n+1)L\}) \\
    &=\sum_{n=1}^\ii n P(\{ nL \leq Z < (n+1)L\}) + L
\ead 
It follows that $\sum_{n=1}^\ii nP(\{nL \leq Z < (n+1)L \}) = \ii$. Then:
\bad 
\ii &= \sum_{n=1}^\ii nP(\{nL \leq Z < (n+1)L \}) \\ 
    &= \sum_{n=1}^\ii \sum_{m=1}^n P(\{nL \leq Z < (n+1)L \})\\ 
    &= \sum_{m=1}^\ii \sum_{n=m}^\ii P(\{nL \leq Z < (n+1)L \})\\
    &= \sum_{m=1}^\ii P(\{Z \geq  mL\}).
\ead 
\end{proof}

\begin{corol}
Let $\{Z_n \}$ to be a sequence of independent and identically distributed non-negative random variables with infinite mean. Then, for any given $L > 0$, the number of the random variables $Z_n$ that satisfies $Z_n \geq nL$ is infinite almost surely. \thlabel{corol:infinite}
\end{corol}
\begin{proof}
We define the events $E_n = \{\o: Z_n(\o) \geq nL \}$. The corollary then follows from \thref{lem:probsum} and the second Borel-Cantelli lemma.
\end{proof}
\begin{proof}[Proof of \thref{thm:infinite-jump}] We first define the random variable $s'_{n} = \ff{1}{n-1}\sum_{i=1} \lp X_i - \m \rp ^2$ for $n\geq 2$. Next, we are going to show that, almost surely, infinitely many elements of the sequence $\{s'_n\}$ satisfy $s'_n \geq L$. Let $T_n  = (X_n-\m)^2$. Then, \thref{corol:infinite} applies and there are infinitely many $T_n$ that satisfies $T_n \geq nL$. Let $n_k \geq 2$ for $k \geq 1$ be an increasing sequence that satisfies $T_{n_k} \geq n_kL$. Then for each $k \geq 1$, we have $s'_{n_k} \geq \ff{1}{n_k-1} T_{n_k} > L$. Let $\O_{s',L}$ be the set of outcomes such that $\{s'_n > L\}$ is satisfied for infinitely many $n$. We showed that $P(\O_{s',L}) = 1$.

We now show that a similar statement holds for $s_n = \ff{1}{n-1}\sum_{i=1}^n \lp X_i - \bb X_n \rp^2$. By the SLLN, there is a set of outcomes $\O_{\bb X\to \m}$ such that $\lim_{n\to \ii} \bb X_n(\o) = \m$ for all $\o \in \O_{\bb X_n \to \m}$ and $P(\O_{\bb X_n\to \m})=1$. Let $\O_s = \O_{s',3L} \cap \O_{\bb X \to \m}$. Choose an arbitrary $\o \in \O_s$. Since $\O_s \subseteq \O_{s',3L}$, there is an infinite sequence $\{n_k \geq 2\}$ such that $s'_{n_k}(\o) \geq 3L$. As $\bb X_n$ converges to $\m$ in $\O_{s}\subseteq\O_{\bb X_n \to \m}$, the sequence $\{n_k\}$ has an infinite subsequence $\{m_k \geq 2\}$ that also satisfies $\abs{\bb X_{m_k}(\o)-\m} < L$. As $s_n(\o) = s'_n(\o) - \ff{n}{n-1} \lp X_n(\o)-\m \rp^2$, it follows that $s'_{m_k} \geq \ff{2n-3}{n-1}L$. Since for $n \geq 2$ $2n-3 \geq n-1$ is satisfied, $s'_{m_k} \geq L$ is valid for all $m_k$. 
The theorem is proved if we can show $P(\O_{s'}) = 1$. To see this note that $P(\O_{s}) = 1 \iff P(\O \setminus \O_{s}) = 0$. The latter follows from the following relation. $P(\O \setminus \O_{s}) = P\lp \O \cap \lp \O^\mathsf{c}_{s',3L} \cup \O^\mathsf{c}_{\bb X_n \to \m}\rp\rp  \leq P(\O \cap \O_{\bb X_n \to \m}^\mathsf{c})+P(\O \cap \O_{s'}^\mathsf{c}) = 0$. 
\end{proof}

Let $\O$ be a finite sample space associated with the $\s$-algebra $\calf = 2^\O$,the set of all subsets of $\O$, and a family of probability distributions $P^t:\calf \to [0,1]$ for $t \in (0,1]$. We assume that $P^t$ is continuous in the sense that $P^t(\o)$ is a continuous function of $t$ for $t \in (0,1]$ for all $\o \in \O$. We consider a non-negative random variable $X^t$ which is continuous in $t$ in the same sense. We further assume that there is a set $E \subset \O$ such that $\lim_{t \to 0} P^t(\o) = 0$ and $\lim_{t \to 0} P^t(\o)X^t(\o) \neq 0$ for all $\o \in E$. 

\begin{thm}
Let $\d,\e > 0$. There is an integer $N(\d,\e)$ such that for all $N \geq N(\d,\e)$:
\be 
\lim_{t\to 0} P^t(\abs{\bb X^t_{N} - \lp \m-\Delta \rp } \leq \d  ) \geq  1-\e.
\label{eq:thmgap}
\ee 
\thlabel{thm:gap}

\begin{proof}[Proof of \thref{thm:gap}]
We first define another probability measure on $\O$ that we will denote by $P^0$. $P^0$ is defined by $P^0(\o) = \lim_{t \to 0} P^t(\o)$ for all $\o \in \O$. We also define $X^0$ similarly: $X^0(\o) = \lim_{t \to 0} X^t(\o)$. Effectively, this definition ignores exceptional configurations. It follows that, expectation value of $X^0$ is $\m-\D$:
\begin{widetext}
\bad 
\m_{X^0} &= \sum_{\o \in \O}P^0(\o)X^0(\o)\\ 
&= \sum_{\o \in E}P^0(\o)X^0(\o) + \sum_{\o \in \lp \O \setminus E\rp} P^0(\o)X^0(\o) \\ 
&= \sum_{\o \in \lp \O \setminus E\rp} \lim_{t\to 0}P^t(\o)\lim_{t\to 0}X^t(\o) \\
&= \lim_{t\to 0}\sum_{\o \in \lp \O \setminus E\rp} P^t(\o)X^t(\o) \\ 
&= \m-\D.
\ead 
Now given $\d,\e>0$, by the WLLN there is an integer $N(\d,\e)$ such that for all $N\geq N(\d,\e)$:
\be
P \lp \abs{\bb X^0_{N}-(\m-\D)} \leq \d \rp \geq 1-\e .
\ee 
Now we consider $\lp E^\mathsf{c} \rp ^N$, the set of ensembles of sample size $N$ that does not include any exceptional configurations where $E^\mathsf{c} \Omega \setminus E$. For $P^0$, exceptional configurations can be ignored and therefore $\O^{N} \equiv \lp E^\mathsf{c} \rp^N$ effectively so it follows that:
\bad
P^0\lp \abs{\bb X^0_N - \lp \m-\D \rp } \leq \d \rp &= P^0 \lp \abs{\bb X^0_N - \lp \m-\D \rp} \leq \d \bigg| \o^N \in \lp E^\mathsf{c} \rp^N \rp \\ 
&= \lim_{t\to 0} P^t \lp \abs{\bb X^t_N - \lp \m-\D \rp} \leq \d  \bigg| \o^N \in \lp E^\mathsf{c} \rp^N \rp.
\label{eq:limit} 
\ead

Now we make the following observation. Let $\lp \lp E^\mathsf{c} \rp ^{N}\rp^\mathsf{c} \subset \O^{N}$ be the subset of $\O^{N_S}$ that includes at least one element from $E$, the set of the exceptional configurations. The probability of $\lp \lp E^\mathsf{c} \rp ^N\rp^\mathsf{c}$ occurs is a polynomial in the variables $\lbrace P^t(\o): \o \in E \rbrace$ with the constant term is vanishing. Since $\lim_{t\to 0}P^t(\o)=0$ for all $\o \in E$, we have:
\bad 
\lim_{t\to 0}P^t \lp \lp \lp E^\mathsf{c} \rp ^N\rp^\mathsf{c} \rp &= 0 \\ 
\lim_{t\to 0}P^t \lp \lp E^\mathsf{c} \rp ^N\rp &= 1 \label{eq:dense}
\ead 

Now we complete the proof of Theorem \ref{thm:gap} by combining Eqs.~\eqref{eq:limit} and \eqref{eq:dense}:
\bad 
\lim_{t\to 0} P^t \lp \abs{\bb X^t_N - \lp \m-\D \rp} \leq \d \rp &= \lim_{t\to 0} P^t \lp \abs{\bb X^t_N - \lp \m-\D \rp} \leq \d  \bigg| \o^N \in \lp E^\mathsf{c} \rp^N \rp P^t \lp \o^N \in  \lp E^\mathsf{c}\rp^N  \rp \\ 
& + \lim_{t\to 0} P^t \lp \abs{\bb X^t_N - \lp \m-\D \rp} \leq \d  \bigg| \o^N \in \lp \lp E^\mathsf{c} \rp ^N\rp^\mathsf{c}\rp P^t \lp \o^N \in \lp \lp E^\mathsf{c} \rp ^N\rp^\mathsf{c} \rp \\ 
&= \lim_{t\to 0} P^t \lp \abs{\bb X^t_N - \lp \m-\D \rp} \leq \d  \bigg| \o^N \in \lp E^\mathsf{c} \rp^N \rp \\
&= P^0\lp \abs{\bb X^0_N - \lp \m-\D \rp } \leq \d \rp \\ 
&\geq 1-\e.
\ead
\end{widetext}
\end{proof}
\end{thm} 

\section{Median of Means} \label{appendix:mom}
In this section we will prove Eq.~\eqref{eq:mom_confidence} for the median of means by modifying the arguments given in Ref.~\cite{lerasle2019lecture} to include correlations between samples. Consider a random variable $X$ with mean $\mu_X$. Given an $\e > 0$, we aim to find a lower bound for the probability that $\abs{\hat \m_\mom - \m_X } < \e$, where $\hat \m_\mom$ is defined in Sec.~\ref{sec:non-asymptotic-estimator}. If there are $K$ batches of size $B$, for this to happen less than $\ff K 2$ of the batch means $\hat \m_i$ must be outside the range $(\m_X-\e,\m_X+\e)$. Let us define the indicator random variables $I_i$ for $i=1,\cds, K$. $I_i$ defined to be $1$ if $\hat \m_i \in (\m-\e,\m+\e)$ and $0$ otherwise. Consequently:
\be 
Prob\lp \abs{\hat \m_\mom-\m} < \e \rp \geq Prob\lp \ff 1 K \sum_{i=1}^K {I_i} < \ff 1 2 \rp 
\ee
Since the batches are independent, we can use Hoeffding's inequality (\cite{durrett_2019})
\be 
Prob\lp \ff 1 K \sum_{i=1}^K {I_i} < \ff 1 2 \rp \geq 1-e^{-2K\lp \ff 1 2 - \mathbb{E}[I_1] \rp^2},
\ee 
where the first indicator function $I_1$ is chosen for convenience. Now we define $\hat \m_1$ to be the standard deviation of $X$ and use Chebyshev's inequality (\cite{durrett_2019}) to obtain:
\bad 
\mathbb{E}[I_1] - \ff 1 2 &= \ff 1 2  - Prob\lp \abs{\hat \m_1 - \m} \geq \e \rp\\
&\geq \ff{1}{2} - \ff{\s_1^2}{\e^2}.
\ead 
By choosing $\e = 2\s_1$, we obtain:
\be 
Prob\lp \abs{\hat \m_\mom-\m} < 2\s_1 \rp \geq 1 - e^{-\ff K 8} \label{eq:mom_helper}
\ee 
To estimate $\s_1$, we note that $\hat \m_1 = \ff 1 B \sum_{n=1}^B X_n$. Then one obtains:
\bad 
\s_1^2 &= \ff{1}{B^2} \lk \sum_{n=1}^B Var(X_n) + 2\sum_{m<n} Cov(X_m,X_n)\rk \\ 
&=\ff{1}{B^2} \lk B \s^2 + 2 \sum_{m=1}^{B-1} \sum_{n=m+1}^{B} Cov(X_m,X_n) \rk \\ 
&= \ff{1}{B^2} \lk B \s^2 + 2\s^2 \sum_{m=1}^{B-1} \sum_{n=m+1}^{B} \G_X(n-m) \rk \\
&= \ff{1}{B^2} \lk B \s^2 + 2\s^2 \sum_{t=1}^{B-1} \lp B-t \rp \G_X(t) \rk \\ 
&= \ff{\s^2}{B}\lk 1 + 2\sum_{t=1}^{B-1}\lp 1-\ff{t}{B}\rp  \G_X(t)\rk \\ 
&= \ff{\s^2}{B}2\tau_{X,int}(B),
\ead
where we have defined the autocorrelation function $\G_X(t) \equiv \ff{1}{\s^2} Cov(X_n,X_{n+t})$ and the integrated autocorrelation time $\tau_{X,int}(B) = \ff 1 2 + \sum_{t=1}^{B-1}(1-\ff{t}{B})  \G_X(t)$ ( the sequence $\{X_n\}$ is assumed to be stationary, $\G_X(t)$ is independent of $n$.). Eq.~\eqref{eq:mom_confidence} then follows by combining above inequality with \eqref{eq:mom_helper}.

\section{Gauss-Hermite Quadrature}
\label{app:hermite}
The polynomials $\he_n(\xi)$ are defined by:
\be 
\he_n(\xi) = (-1)^n e^{\ff 1 2 \xi^2}\frac{d^n}{dt^n}e^{-\ff 1 2 \xi^2}
\ee 
and have the properties:
\ba
\frac{1}{\ss {2\pi}} \int dt\, e^{-\ff 1 2 \xi^2} \xi^m \he_n(\xi) &= 0 \quad \text{ for } 0\leq m<n, \label{eq:quadrature-vanishing} \\
\frac{1}{\ss {2\pi}} \int dt\, e^{-\ff 1 2 \xi^2} \he_m(\xi) \he_n(\xi) &= n! \d_{nm}, \label{eq:orthogonality} \\ 
H_n(\xi) &= \xi^n + \cds. \label{eq:hermite-form}
\ea
Consider a polynomial $f(\xi)$ of degree at most $2n-1$. We can then write 
\be
f(\xi) = q(\xi)\he_n(\xi)+r(\xi)
\ee
where both $q(\xi)$ and $r(\xi)$ have degree at most $n-1$. By Eq.~\eqref{eq:quadrature-vanishing} one has:
\be 
\frac{1}{\ss {2\pi}} \int d\xi\, e^{-\xi^2/2} q(\xi)\he_n(\xi) = 0,
\ee 
and therefore one obtains:
\be 
\frac{1}{\ss {2\pi}} \int dt\, e^{-\ff 1 2 \xi^2}f(\xi) = \frac{1}{\ss {2\pi}} \int d\xi\, e^{-\ff 1 2 \xi^2} r(\xi).
\label{eq:quadrature-reduction}
\ee

As $r(\xi)$ is a polynomial of degree at most $n-1$, it is determined by its values at $n$ points. Let us choose these points $\left\{\xi^{(n)}_a\rvert a = 1,\cds n \right\}$ as the roots of $\he_n(\xi)=\prod_{a=1}^n(\xi -\xi_a^{(n)})$. Then we have:
\bad 
r(\xi) &= \sum_a r(\xi_a) \prod_{\ms{1 \leq b \leq n \\ b \neq a}} \ff{\xi-\xi_b}{\xi_a-\xi_b} \\ &= \sum_a r(\xi_a)\frac{1}{\prod_{b\neq a} (\xi_a-\xi_b)}\ff{\he_n(\xi)}{\xi-\xi_a}.
\ead
This allows us to express Eq.~\eqref{eq:quadrature-reduction} as:
\bad 
\frac{1}{\ss {2\pi}} \int dt\, e^{-\ff 1 2 \xi^2}r(\xi) &= \sum_a r(\xi_a)\frac{1}{\prod_{b\neq a} (\xi_a-\xi_b)} \\ &\times \frac{1}{\ss {2\pi}} \int dt\, e^{-\ff 1 2 \xi^2}\ff{\he_n(\xi)}{\xi-\xi_a}.
\ead 
After defining $w_a$ as:
\be 
w_a =  \ff{1}{\prod_{b\neq a}\lp \xi_a-\xi_b \rp}\frac{1}{\ss {2\pi}} \int dt\, e^{-\ff 1 2 \xi^2}\ff{\he_n(\xi)}{\xi-\xi_a},
\ee 
one can use $f(\xi_a) = q(\xi_a)\he_n(\xi_a)+r(\xi_a) = r(\xi_a)$ to obtain:
\be 
\frac{1}{\ss {2\pi}} \int d\xi\, e^{-\ff 1 2 \xi^2}f(\xi) = \sum_a w_a f(\xi_a).
\ee

Moreover, if $q(\xi)$ is a polynomial of degree at most $n$:
\bad
\frac{1}{\prod_{b\neq a} (\xi_a-\xi_b)} \frac{1}{\ss {2\pi}} \int dt\, e^{-\ff 1 2 \xi^2}\ff{\he_n(\xi)}{\xi-\xi_a}q(\xi) = q(\xi_a)w_a.
\label{eq:hermite-qw-rel}
\ead
To see that consider $q(\xi) = \xi^m$ for $m < n$. Using the identity 
\bad 
\xi^m = \xi_a^m + \lp \xi-\xi_a\rp \lp \xi_a^{m-1}+\xi_a^{m-2}\xi+\cds+\xi_a\xi^{m-2}+\xi^{m-1}\rp
\ead
and Eq.~\eqref{eq:quadrature-vanishing}, we see that
\bad
\frac{1}{\prod_{b\neq a} (\xi_a-\xi_b)} \frac{1}{\ss {2\pi}} \int dt\, e^{-\ff 1 2 \xi^2}\ff{\he_n(\xi)}{\xi-\xi_a}\xi^m = \xi_a^m w_a
\ead
for $m\leq n$, from which \eqref{eq:hermite-qw-rel} follows.
Now choosing $q(\xi) = \he_{n-1}(\xi)$ and using Eqs.~\eqref{eq:quadrature-vanishing} and \eqref{eq:hermite-form}, the left hand side of Eq.~\eqref{eq:hermite-qw-rel}\ becomes:
\bad 
\frac{1}{\prod_{b\neq a} (\xi_a-\xi_b)}\frac{1}{\ss {2\pi}} \int dt\, e^{-\ff 1 2 \xi^2}\he_{n-1}(\xi)\he_{n-1}(\xi).
\ead 
Using Eq.~\eqref{eq:orthogonality}, this leads to:
\bad
w_a = \ff{(n-1)}{\he'_{n}(\xi_a)\he_{n-1}(\xi_a)},
\ead 
where $w_a$ satisfies the normalization:
\bad 
\sum_a w_a = 1.
\ead 

A table of the weights and roots corresponding to the Hermite polynomials used in this work is provided in Table~\ref{tab:hermite-table-all}.
\begin{table}[!th]
\begin{tabular}{ |c|c|c| }
\hline
n & Roots $\xi_a^{(n)}$ & Weights $w_a^{(n)}$ \\ \hline
\multirow{1}{*}{2} 
 & $1$ & $1/2$ \\
 \hline
\multirow{2}{*}{3} 
 & $0$ & $2/3$ \\
 & $\ss 3$ & $1/6$ \\
\hline
\multirow{2}{*}{4}
 & $\ss{3-\ss{6}}$ & $1/12 \lp 3 +\ss 6 \rp$ \\
 & $\ss{3+\ss 6}$ & $1/12 \lp 3 - \ss 6 \rp$ \\ 
\hline
\multirow{3}{*}{5}
& $ 0 $ & $ 0.5333333333 $ \\
& $ 1.3556261800 $ & $ 0.2220759220 $ \\
& $ 2.8569700139 $ & $ 0.0112574113 $ \\
\hline
\multirow{3}{*}{6}
& $ 0.6167065902 $ & $ 0.4088284696 $ \\
& $ 1.8891758778 $ & $ 0.0886157460 $ \\
& $ 3.3242574336 $ & $ 0.0025557844 $ \\
\hline
\multirow{4}{*}{7}
& $ 0 $ & $ 0.4571428571 $ \\
& $ 1.1544053947 $ & $ 0.2401231786 $ \\
& $ 2.3667594107 $ & $ 0.0307571240 $ \\
& $ 3.7504397177 $ & $ 0.0005482689 $ \\
\hline
\multirow{4}{*}{8}
& $ 0.5390798114 $ & $ 0.3730122577 $ \\
& $ 1.6365190424 $ & $ 0.1172399077 $ \\
& $ 2.8024858613 $ & $ 0.0096352201 $ \\
& $ 4.1445471861 $ & $ 0.0001126145 $ \\
\hline
\multirow{5}{*}{9}
& $ 0 $ & $ 0.4063492063 $ \\
& $ 1.0232556638 $ & $ 0.2440975029 $ \\
& $ 2.0768479787 $ & $ 0.0499164068 $ \\
& $ 3.2054290029 $ & $ 0.0027891413 $ \\
& $ 4.5127458634 $ & $ 0.0000223458 $ \\
\hline
\end{tabular}
\caption{Non-negative roots and weights for the first few Hermite polynomials. For every positive root, there is a negative root with the same magnitude and weight.}
\label{tab:hermite-table-all}
\end{table}

\FloatBarrier
\bibliography{refs}

\begin{thebibliography}{11}%
\makeatletter
\providecommand \@ifxundefined [1]{%
 \@ifx{#1\undefined}
}%
\providecommand \@ifnum [1]{%
 \ifnum #1\expandafter \@firstoftwo
 \else \expandafter \@secondoftwo
 \fi
}%
\providecommand \@ifx [1]{%
 \ifx #1\expandafter \@firstoftwo
 \else \expandafter \@secondoftwo
 \fi
}%
\providecommand \natexlab [1]{#1}%
\providecommand \enquote  [1]{``#1''}%
\providecommand \bibnamefont  [1]{#1}%
\providecommand \bibfnamefont [1]{#1}%
\providecommand \citenamefont [1]{#1}%
\providecommand \href@noop [0]{\@secondoftwo}%
\providecommand \href [0]{\begingroup \@sanitize@url \@href}%
\providecommand \@href[1]{\@@startlink{#1}\@@href}%
\providecommand \@@href[1]{\endgroup#1\@@endlink}%
\providecommand \@sanitize@url [0]{\catcode `\\12\catcode `\$12\catcode
  `\&12\catcode `\#12\catcode `\^12\catcode `\_12\catcode `\%12\relax}%
\providecommand \@@startlink[1]{}%
\providecommand \@@endlink[0]{}%
\providecommand \url  [0]{\begingroup\@sanitize@url \@url }%
\providecommand \@url [1]{\endgroup\@href {#1}{\urlprefix }}%
\providecommand \urlprefix  [0]{URL }%
\providecommand \Eprint [0]{\href }%
\providecommand \doibase [0]{http://dx.doi.org/}%
\providecommand \selectlanguage [0]{\@gobble}%
\providecommand \bibinfo  [0]{\@secondoftwo}%
\providecommand \bibfield  [0]{\@secondoftwo}%
\providecommand \translation [1]{[#1]}%
\providecommand \BibitemOpen [0]{}%
\providecommand \bibitemStop [0]{}%
\providecommand \bibitemNoStop [0]{.\EOS\space}%
\providecommand \EOS [0]{\spacefactor3000\relax}%
\providecommand \BibitemShut  [1]{\csname bibitem#1\endcsname}%
\let\auto@bib@innerbib\@empty
\bibitem [{\citenamefont {Shi}\ and\ \citenamefont {Zhang}(2016)}]{inf_var_mc}%
  \BibitemOpen
  \bibfield  {author} {\bibinfo {author} {\bibfnamefont {H.}~\bibnamefont
  {Shi}}\ and\ \bibinfo {author} {\bibfnamefont {S.}~\bibnamefont {Zhang}},\
  }\href {\doibase 10.1103/physreve.93.033303} {\bibfield  {journal} {\bibinfo
  {journal} {Physical Review E}\ }\textbf {\bibinfo {volume} {93}} (\bibinfo
  {year} {2016}),\ 10.1103/physreve.93.033303}\BibitemShut {NoStop}%
\bibitem [{\citenamefont {{Gross}}\ and\ \citenamefont
  {{Neveu}}(1974)}]{original}%
  \BibitemOpen
  \bibfield  {author} {\bibinfo {author} {\bibfnamefont {D.~J.}\ \bibnamefont
  {{Gross}}}\ and\ \bibinfo {author} {\bibfnamefont {A.}~\bibnamefont
  {{Neveu}}},\ }\href {\doibase 10.1103/PhysRevD.10.3235} {\bibfield  {journal}
  {\bibinfo  {journal} {Physical Review D}\ }\textbf {\bibinfo {volume} {10}},\
  \bibinfo {pages} {3235} (\bibinfo {year} {1974})}\BibitemShut {NoStop}%
\bibitem [{\citenamefont {{Hubbard}}(1959)}]{hubbard}%
  \BibitemOpen
  \bibfield  {author} {\bibinfo {author} {\bibfnamefont {J.}~\bibnamefont
  {{Hubbard}}},\ }\href {\doibase 10.1103/PhysRevLett.3.77} {\bibfield
  {journal} {\bibinfo  {journal} {\prl}\ }\textbf {\bibinfo {volume} {3}},\
  \bibinfo {pages} {77} (\bibinfo {year} {1959})}\BibitemShut {NoStop}%
\bibitem [{\citenamefont {{Stratonovich}}(1957)}]{stratonovich}%
  \BibitemOpen
  \bibfield  {author} {\bibinfo {author} {\bibfnamefont {R.~L.}\ \bibnamefont
  {{Stratonovich}}},\ }\href@noop {} {\bibfield  {journal} {\bibinfo  {journal}
  {Soviet Physics Doklady}\ }\textbf {\bibinfo {volume} {2}},\ \bibinfo {pages}
  {416} (\bibinfo {year} {1957})}\BibitemShut {NoStop}%
\bibitem [{\citenamefont {Wilson}(1974)}]{wilson-fermions}%
  \BibitemOpen
  \bibfield  {author} {\bibinfo {author} {\bibfnamefont {K.~G.}\ \bibnamefont
  {Wilson}},\ }\href {\doibase 10.1103/PhysRevD.10.2445} {\bibfield  {journal}
  {\bibinfo  {journal} {Phys. Rev. D}\ }\textbf {\bibinfo {volume} {10}},\
  \bibinfo {pages} {2445} (\bibinfo {year} {1974})}\BibitemShut {NoStop}%
\bibitem [{\citenamefont {Hirsch}(1983)}]{hirsch}%
  \BibitemOpen
  \bibfield  {author} {\bibinfo {author} {\bibfnamefont {J.~E.}\ \bibnamefont
  {Hirsch}},\ }\href {\doibase 10.1103/PhysRevB.28.4059} {\bibfield  {journal}
  {\bibinfo  {journal} {Phys. Rev. B}\ }\textbf {\bibinfo {volume} {28}},\
  \bibinfo {pages} {4059} (\bibinfo {year} {1983})}\BibitemShut {NoStop}%
\bibitem [{\citenamefont {Bulgac}\ \emph {et~al.}(2008)\citenamefont {Bulgac},
  \citenamefont {Drut},\ and\ \citenamefont {Magierski}}]{Drut}%
  \BibitemOpen
  \bibfield  {author} {\bibinfo {author} {\bibfnamefont {A.}~\bibnamefont
  {Bulgac}}, \bibinfo {author} {\bibfnamefont {J.~E.}\ \bibnamefont {Drut}}, \
  and\ \bibinfo {author} {\bibfnamefont {P.}~\bibnamefont {Magierski}},\ }\href
  {\doibase 10.1103/physreva.78.023625} {\bibfield  {journal} {\bibinfo
  {journal} {Physical Review A}\ }\textbf {\bibinfo {volume} {78}} (\bibinfo
  {year} {2008}),\ 10.1103/physreva.78.023625}\BibitemShut {NoStop}%
\bibitem [{\citenamefont {Wu}\ and\ \citenamefont {Zhang}(2005)}]{Wu:2005zzb}%
  \BibitemOpen
  \bibfield  {author} {\bibinfo {author} {\bibfnamefont {C.}~\bibnamefont
  {Wu}}\ and\ \bibinfo {author} {\bibfnamefont {S.-C.}\ \bibnamefont {Zhang}},\
  }\href {\doibase 10.1103/PhysRevB.71.155115} {\bibfield  {journal} {\bibinfo
  {journal} {Phys. Rev. B}\ }\textbf {\bibinfo {volume} {71}},\ \bibinfo
  {pages} {155115} (\bibinfo {year} {2005})},\ \Eprint
  {http://arxiv.org/abs/cond-mat/0407272} {arXiv:cond-mat/0407272} \BibitemShut
  {NoStop}%
\bibitem [{\citenamefont {Wolff}\ \emph {et~al.}(2004)\citenamefont {Wolff},
  \citenamefont {Collaboration} \emph {et~al.}}]{wolff2004monte}%
  \BibitemOpen
  \bibfield  {author} {\bibinfo {author} {\bibfnamefont {U.}~\bibnamefont
  {Wolff}}, \bibinfo {author} {\bibfnamefont {A.}~\bibnamefont
  {Collaboration}},  \emph {et~al.},\ }\href@noop {} {\bibfield  {journal}
  {\bibinfo  {journal} {Computer Physics Communications}\ }\textbf {\bibinfo
  {volume} {156}},\ \bibinfo {pages} {143} (\bibinfo {year}
  {2004})}\BibitemShut {NoStop}%
\bibitem [{\citenamefont {Durrett}(2019)}]{durrett_2019}%
  \BibitemOpen
  \bibfield  {author} {\bibinfo {author} {\bibfnamefont {R.}~\bibnamefont
  {Durrett}},\ }\href {\doibase 10.1017/9781108591034} {\emph {\bibinfo {title}
  {Probability: Theory and Examples}}},\ \bibinfo {edition} {5th}\ ed.,\
  Cambridge Series in Statistical and Probabilistic Mathematics\ (\bibinfo
  {publisher} {Cambridge University Press},\ \bibinfo {year}
  {2019})\BibitemShut {NoStop}%
\bibitem [{\citenamefont {Lerasle}(2019)}]{lerasle2019lecture}%
  \BibitemOpen
  \bibfield  {author} {\bibinfo {author} {\bibfnamefont {M.}~\bibnamefont
  {Lerasle}},\ }\href@noop {} {\enquote {\bibinfo {title} {Lecture notes:
  Selected topics on robust statistical learning theory},}\ } (\bibinfo {year}
  {2019}),\ \Eprint {http://arxiv.org/abs/1908.10761} {arXiv:1908.10761
  [stat.ML]} \BibitemShut {NoStop}%
\end{thebibliography}%

\end{document}